\documentclass[12pt]{article}

\RequirePackage[OT1]{fontenc}
\usepackage{amsthm,amsmath,natbib}
\RequirePackage{hypernat}
\usepackage[pdftex]{graphicx}
\usepackage{amsfonts}
\usepackage{epic}
\usepackage{fancybox}
\usepackage[pdftex]{color}
\usepackage{marginnote}
\usepackage{enumerate}
\usepackage{multirow}

\usepackage[letterpaper]{geometry}
\RequirePackage{ifthen}

\usepackage{amsmath}

\usepackage{graphicx}
\usepackage{geometry}
\usepackage[utf8]{inputenc}
\usepackage{amsfonts,epsfig}
\usepackage[hyphens]{url}
\usepackage{hyperref}
\usepackage{multirow}
\usepackage[linesnumbered,ruled,vlined]{algorithm2e}
\usepackage{amssymb,amsmath}
\usepackage{mathtools}
\graphicspath{{plots/}}

\newtheorem{alg}{Algorithm}
\newtheorem{theorem}{Theorem}

\newtheorem{lemma}{Lemma}
\newtheorem{lemma*}{Lemma}
\newtheorem{remark}{Remark}
\newtheorem{corollary}{Corollary}
\newtheorem{definition}{Definition}
\newtheorem{example}{Example}

\DeclareMathOperator*{\argmin}{arg\,min}

\newboolean{inline}
\setboolean{inline}{true}

\title{
Finite Sample Valid Inference via Calibrated Bootstrap
}

\author{%
Yiran Jiang\textsuperscript{1}, %
Chuanhai Liu\textsuperscript{2}\textsuperscript{*}, %
Heping Zhang\textsuperscript{1} %
}

\date{\today}

\begin{document}

\maketitle
\vspace{-0.2 in}
\begin{center}
\textsuperscript{1}Department of Biostatistics, Yale University, New Haven, USA\\
\textsuperscript{2}Department of Statistics, Purdue University, West Lafayette, USA\\
\textsuperscript{*}Email: chuanhai@purdue.edu
\end{center}

\vspace{0.2 in}

\begin{abstract}
While widely used as a general method for uncertainty quantification, the bootstrap method encounters difficulties that raise concerns about its validity in practical applications. This paper introduces a new resampling-based method, termed {\it calibrated bootstrap}, designed to generate finite sample-valid parametric inference from a sample of size $n$. The central idea is to calibrate an $m$-out-of-$n$ resampling scheme, where the calibration parameter $m$ is determined against inferential pivotal quantities derived from the cumulative distribution functions of loss functions in parameter estimation. The method comprises two algorithms. The first, named {\it resampling approximation} (RA), employs a {\it stochastic approximation} algorithm to find the value of the calibration parameter $m=m_\alpha$ for a given $\alpha$ in a manner that ensures the resulting $m$-out-of-$n$ bootstrapped $1-\alpha$ confidence set is valid. The second algorithm, termed {\it distributional resampling} (DR),  is developed to further select samples of bootstrapped estimates from the RA step when constructing $1-\alpha$ confidence sets for a range of $\alpha$ values is of interest. The proposed method is illustrated and compared to existing methods using linear regression with and without $L_1$ penalty, within the context of a high-dimensional setting and a real-world data application. The paper concludes with remarks on a few open problems worthy of consideration.
\end{abstract}

\section{Introduction}



Bootstrap methods are designed to estimate the sample distribution of the statistic of interest by resampling the observed
data. Since the seminal paper of \cite{efron1979bootstrap} \citep[see also][]{efron2003second}, these methods have often served as an indispensable tool in the realm of statistical inference and prediction, due to their simplicity and efficiency \cite[see, {\it e.g.}][]{efron1993introduction}. 
It has also motivated relevant research in other contexts, including its Bayesian counterpart \citep[see][and references therein]{rubin1981bayesian,efron1982jackknife,newton1994approximate,efron2012bayesian,newton2021weighted}.

Nevertheless, practitioners may misuse the bootstrap methods in situations where no theoretical confirmation has been made \citep{shao2012jackknife}.
Care must be taken in specific applications \citep{martin2015plausibility}. Firstly, the theoretical underpinnings of the bootstrap methods predominantly center around its asymptotic validity. Specifically, researchers primarily concentrate on establishing the property of consistent estimation for the sampling distribution of the statistic of interest as the sample size approaches infinity
\citep[see, {\it e.g.},][ and references therein]{bickel1981some,singh1981asymptotic, wellner1996bootstrapping,vanderVaart1996weak,kosorok2008introduction,efron1993introduction}.
Secondly, it has been reported that bootstrap can fail, even in its intended applications based on large-sample theory \citep[see, {\it e.g.,}][and references therein]{bickel1983bootstrapping,abrevaya2005bootstrap, chernozhukov2023high}. For example, the bootstrap has been shown to provide an inconsistent estimation of the true sampling distribution of the coefficients in high-dimensional linear regression setting of $n \to \infty$ while $p/n \to c$ for some $c > 0$ with $p$ predictors, and the numerical experiments also show that the existing bootstrap methods give very poor inference on the vector of regression coefficients $\beta$ \citep{el2018can}. 

To overcome the limitations of the standard bootstrap method ({\it i.e.}, the bootstrap with $n$-out-of-$n$ replacement resampling) in certain scenarios, researchers have explored alternative resampling approaches. These alternatives, also required to be supported by their asymptotic validity, include $m$-out-of-$n$ resampling with replacement and without replacement, where the latter is also known as subset sampling \citep{politis1994large,bickel2008choice,bickel2012resampling}, an idea that dates back to \cite{bretagnolle1983lois}. 
However, to the best of our knowledge,  there is still a significant gap in the existing research when it comes to the development of theoretically supported finite sample inference methods based on resampling. 

In order to achieve desirable finite-sample performance, the resampling scheme may need to adapt itself according to the specific model and observations.
In this paper, we explore adaptive and computationally feasible resampling methods for achieving valid frequentist inference for model parameters in the finite sample case. 
That is, for a given model, the inference method should be able to provide control over the type-I error rates. 
For example, the constructed $95\%$ confidence interval of the parameter should cover the true parameter at least (and close to) 95\% of the time.   For this, we develop computational methods with sound theoretical justifications, along the lines of recent developments on the foundations of statistical inference \citep{fisher1973statistical,shafer1976mathematical,dempster2008dempster,singh2007confidence,xie2013confidence,hannig2009generalized,hannig2016generalized,martin2013inferential,martin2015inferential,martin2015plausibility,cella2022direct}, while our exposition will make use of the basic idea of the familiar pivotal method as much as possible. 

Theoretically sound methods for finite-sample valid inference have been well-studied in the existing literature. However, the computational tools that can be practically applied to find solutions are generally lacking and often infeasible, especially in high-dimensional scenarios \citep{martin2015plausibility}. The primary focus of this paper is to propose a computationally efficient, resampling-based numerical strategy to closely approximate the targeted theoretical solution. Philosophically, it is at least subtle to make inference via resampling because variability from resampling and uncertainty assessment for inference are two different concepts. 
Here, our key idea is to match the variability from resampling to uncertainty in inference. By looking at inferential problems from the 
inferential models (IMs) perspective \citep{martin2013inferential}, which could be considered as rooted in the pivotal method for hypothesis testing and confidence interval construction, we numerically quantify 
uncertainty assessment by considering loss functions in the context of parameter estimation and making use of the cumulative distributions of their sampling distributions to introduce the needed pivotal quantities or auxiliary random variables.
Consequently, we seek resampling schemes in such a way that the distribution of the loss function based on the bootstrapped estimates matches the theoretical distribution necessary for valid inference. 
Because our primary goal is to create a method that leads to valid frequentist inference in finite sample scenarios, regardless of the asymptotic relationship between the sample size $n$ and the number of parameters $p$, the method inherently possesses the capacity to be applicable in high-dimensional settings, for which the standard bootstrap has been known to be problematic as discussed above.

The rest of the paper is arranged as follows. In Section~\ref{s:motivation}, we will briefly review the generalized association method \citep{martin2015plausibility}, a foundational approach for valid parametric inference, which forms the theoretical basis of our proposed method.
In Section \ref{s:SA}, we propose a variant of bootstrap that aims to produce valid frequentist inference in finite sample scenarios. This method, called {\it calibrated bootstrap} (CB),  
consists of two steps. The first step applies a resampling method, called the {\it resampling approximation} (RA), to a set of pre-specified coverage probabilities. For each confidence coefficient, RA performs the $m$-out-of-$n$ resampling with replacement, with the value of $m$ obtained by a {\it stochastic approximation} algorithm to ensure that the resulting confidence region has the desired coverage probability.
The second step of CB, called the {\it distributional resampling} (DR), is optional and selects a common sample of the bootstrapped estimates that works for a pre-specified range of $\alpha$ values of interest. 
Although our exposition emphasizes joint inference on the unknown parameters, we also discuss marginal inference on a parameter of interest. The proposed CB method is illustrated with applications in both easy-to-verify linear regression in Section~\ref{s:SA} and challenging $L_1$ penalized linear regression in Section~\ref{s:application}, including a real data example from a diabetes study.
We conclude with a few remarks in Section \ref{s:discussion}.

\section{Foundations of Valid Inference: an Overview}
\label{s:motivation}

\ifthenelse{1=1}{}{
\subsection{Simple Examples of Valid Inference with the Pivotal Method}
\label{ss:examples} 
In this section, we use two classic examples to explain the key idea of making valid inferences using the pivotal method and its extensions, including fiducial inference \citep{fisher1930inverse} and its successors \citep[see, {\it e.g.},][and references therein]{liu2015frameworks}. 

\begin{example}[Inference on the unknown mean of a normal distribution]\label{example:f-normal}
\rm
Suppose that we have observed a sample of size $n$, $y_1,..., y_n$, from the normal distribution $N(\theta, 1)$ with unknown $\theta \in {\mathbb R}$. We are interested in making fiducial inference on the mean parameter $\theta.$ \CLdel{This can be done through the creation of a \textit{fiducial distribution} of $\theta.$} \CLadd{This can be done using the pivotal method.}

The main idea of fiducial inference can be best explained by the fiducial argument for the one observation case. In this case, $y$ is an observed value from $N(\theta, 1).$ Let
\[
    \gamma = y - \theta.
\]
It can be seen that before observing $y$, the unobserved $\gamma$ follows $N(0, 1)$. From the perspective of \cite{dempster1963further}, the idea of continuing to believe $\gamma\sim N(0, 1)$ after having observed $y$ gives rise to a \textit{fiducial distribution} of $\theta$ as $N(y, 1)$ because $\theta = y - \gamma$. This argument of ``continuing to believe'' forms the basis for fiducial inference. The unobserved quantity $\gamma$ is the so-called pivotal quantity \citep{Weerahandi1993}. 

For the case with $n$ observations, we can introduce the pivotal quantity as $G(\theta, \bar{y}) = \bar{y} - \theta,$ where $\bar{y}$ is the average of the $n$ observations; see \cite{martin2015conditional} for the use of $\bar{y}$ alone for efficient inference. We see that $G(\theta, \bar{y}) \sim N (0, 1/n)$ and, thereby, the fiducial distribution of $\theta$ is the distribution of $\bar{y} - G(\theta, \bar{y})$, which is $N (\bar{y}, 1/n).$ A symmetric $1-\alpha$ confidence interval for $\theta$ is given by $(\bar{y} -\frac{z_{1-\alpha/2}}{\sqrt{n}}, \bar{y} + \frac{z_{1 - \alpha/2}}{\sqrt{n}})$, where $z_{1- \alpha/2}$ is the $1-\alpha/2$ quantile of the standard normal distribution. Thus, the fiducial confidence interval for $\theta$ is the classical $z$-interval, which has the desired coverage probability of $1-\alpha$.

From this simple example, we notice that valid inference is made based on the property of pivotal quantities that their distributions do not depend on the unknown parameters. A reverse operation, such as $\theta=y-\gamma,$ relates the observed data and the pivotal quantity to produce the fiducial distribution for the parameter of interest and the fiducial inference on the parameter. The fiducial distribution can then be used to construct \CLdel{frequentist} confidence intervals. A justification for the validity of these confidence intervals is provided, for example, by \cite{martin2015inferential}.
\end{example}

\begin{remark}
    In fiducial inference, we treat the samples $y$ as random variables when constructing the distribution of the pivotal quantity. However,  we keep $y$ fixed at their observed values when constructing the fiducial distribution. It is crucial to differentiate between these two distinct roles that $y$ plays in this process. In this sense, fiducial inference is also seen as an attempt to obtain prior-free posterior probabilistic inference \citep{martin2013inferential}. Meanwhile, the dual roles of $y$ in this context mirror its function in both density and likelihood functions: as a variable in the former and as fixed observations in the latter.
\end{remark}

\CLadd{
\begin{remark}
    The theoretical foundation of our proposed method is IMs \citep{martin2013inferential}, a framework developed as the successor to traditional fiducial inference. We use fiducial inference and the pivotal method for illustration purposes, due to their conceptual and technical similarities. Consequently, the applicability of our method extends beyond cases where the traditional fiducial distribution is well-defined.
\end{remark}
}

\CLadd{(This example may need to be moved to the supplementary material.)} For the second example, we discuss valid inference about an unknown continuous distribution. Interestingly, we shall see that a simple resampling approximation to a fiducial distribution leads to inference similar to that based on the nonparametric bootstrap.

\begin{example}[Inference on unknown continuous distributions]\label{example:f-nonparametric}
\rm

\hfill{Suppose} that we observe a sample of size $n$, $y_1,..., y_n$, from a univariate continuous distribution $F$ on $\mathbb{R}$. We are interested in making fiducial inference on $F$ or a function of $F.$

Denote by $y_{(1)}, ..., y_{(n)}$ the order statistics of the sample $y_1,..., y_n.$ It follows from what is known as the {\it probability integral transform} that  for all $F$, the values $F(y_{i})$'s form a sample from the $\mbox{Uniform}(0,1)$ distribution. This motivates us to introduce a sample $U_1,..., U_n$ from $\mbox{Uniform}(0,1)$ as the pivotal quantities. Denote by $U_{(1)}= F(y_{(1)}), ..., U_{(n)}=F(y_{(n)})$ the order statistics of the pivotal quantities. The reverse operation produces the set-valued mapping from the space of $U$ to that of $F$:
\begin{equation}\label{eq:fiducial-bootstrap}
        \mathcal{F}_U = \left\{F:\; U_{(i-1)} \leq F(y) \le U_{(i)}\mbox{ for all } y \in [y_{(i-1)}, y_{(i)}),\; i=1,..., n+1\right\},
\end{equation}
where $U_{(0)}= 0$ and $U_{(n+1)}=1.$
For inference on the function $G=g(F)$, we have the fiducial set
\[
    \mathcal{G}_U = \left\{g(F):\; F\in\mathcal{F}_U\right\}.
\]
The fiducial distribution of $F$ can thereby be obtained by treating $U_{(1)},..., U_{(n)}$ in (\ref{eq:fiducial-bootstrap}) as auxiliary random variables. 
Accordingly, to predict a new $y_*$ with the fiducial argument, we can first create a draw of $(U_{(1)}, \dots, U_{(n)})$ and then a draw $U_*$ from $\mbox{Uniform}(0,1)$ to predict $y_*$ by
\begin{equation}\label{eq:fiducial-uniform-13}
    y_* \in [y_{(i-1)}, y_{(i)}), \qquad \qquad  \mbox{ if }
    U_* \in [U_{(i-1)}, U_{(i)});\;\mbox{ for } i=1,...,n,
\end{equation}
where $y_{(0)}=\min \mathcal Y$ and $y_{(n+1)}=\max \mathcal Y$, with $\mathcal{Y}$ representing the support of the distribution $F$.

\begin{remark}
Standard fiducial inference can be invalid \cite[see, e.g.,][and references therein for more discussions on fiducial]{Zabell1992,liu2015frameworks}. However, valid inference (see Definition~\ref{def:valid}) can still be made by predicting the pivotal quantities such as $U_{(1)}, ..., U_{(n)}$ using random sets. Details can be found in \cite{martin2013inferential}. 
\end{remark}


To this end, we consider a sequence of approximations to provide an alternative understanding of nonparametric bootstrap from the perspective of the fiducial inference in this example. First, for simplicity, we approximate the predictive distribution of the set value in \eqref{eq:fiducial-uniform-13} with a usual distribution
by replacing the interval \eqref{eq:fiducial-uniform-13} with the single point $y_{(i)}$ for $i=1,...,n$. That is,
\begin{equation}\label{eq:fiducial-uniform-14}
    y_* = y_{(i)}, \qquad \qquad  \mbox{ if }
    U_* \in [U_{(i-1)}, U_{(i)});\;\mbox{ for } i=1,...,n.
\end{equation}
This type of simplification, {\it i.e.}, replacing set-valued mapping with point-valued mapping, is related to the development in the framework of generalized fiducial inference in \cite{hannig2009generalized} and \cite{hannig2016generalized}.
For practical consideration, we choose to ignore the extreme case at $i = n + 1$. 
We notice that the corresponding probability mass of $y_*$ taking the value $y_{(i)}$ is proportional to $U_{(i)} - U_{(i-1)}$, for $i = 1,\dots,n$. It is known that the uniform spacings $(U_{(1)}-U_{(0)}, ..., U_{(i)}-U_{(i-1)}, ..., U_{(n)}-U_{(n-1)})$ have the same joint distribution as 
\[
  \left(\frac{V_1}{\sum_{k=1}^{n}V_k}, ..., \frac{V_i}{\sum_{k=1}^{n}V_k}, ..., \frac{V_{n}}{\sum_{k=1}^{n}V_k}\right), \qquad V_1, ..., V_{n} \stackrel{iid}{\sim}\mbox{Expo}(1),
\]
where $\mbox{Expo}(1)$ denotes the standard exponential distribution. This leads to the discrete predictive distribution 
\begin{equation}\label{eq:fiducial-uniform-21}
  \mathrm{P}(y_*=y_i)=P_i, \quad P_i = \frac{V_i}{\sum_{k=1}^{n}V_k} \mbox{ with } V_1, ..., V_{n} \stackrel{iid}{\sim}\mbox{Expo}(1);\;\mbox{ for } i=1,...,n.
\end{equation}

Lastly, we show that the bootstrap can be viewed as a way to make inference by approximating \eqref{eq:fiducial-uniform-21} via resampling. That is, the fundamental idea is to approximate the uncertainty represented by \eqref{eq:fiducial-uniform-21} with the variability that is to be provided by resampling. 
Denote by $\{\tilde{y}_1, \dots, \tilde{y}_n\}$ the resampled observations of $\{y_1, \dots, y_n\}$. Thus, the target is for $\{\tilde{y}_1, \dots, \tilde{y}_n\}$ to have approximately the same distribution as that of $\{y_1^*, \dots, y_n^*\}$, where each observation is sampled from the predictive distribution (\ref{eq:fiducial-uniform-21}).
Let $N_i$ be the frequency of the value $y_i$ in the resampled observations $\{\tilde{y}_1, \dots, \tilde{y}_n\}$, for $i = 1, \dots, n$.  Equivalently, we want $(N_1/\sum_{k=1}^nN_k, ..., N_i/\sum_{k=1}^nN_k, ...,N_n/\sum_{k=1}^nN_k)$ to be a rational approximation to $(P_1,..., P_n)$ in (\ref{eq:fiducial-uniform-21}). 
Note that $P=(P_1,..., P_n)$ follows the Dirichlet distribution
\begin{equation}\label{eq:fiducial-uniform-42}
  P \sim \mbox{Dirichlet}\left(\mathbf{1}_{n}\right). 
\end{equation}
Thus, in the corresponding resampling scheme, we allocate the weights generated from $(\ref{eq:fiducial-uniform-42})$  to the observations $\{y_1, \dots, y_n\}$ and resample the observations with these assigned weights. It turns out that this resampling scheme is exactly the Bayesian bootstrap \citep{rubin1981bayesian}. 

Note that
\[
\mbox{E}(P_i) = \frac{1}{n}, \quad\mbox{Var}(P_i) 
= \frac{n-1}{n^2(n+1)}, \quad\mbox{ and }\quad 
\mbox{Cov}(P_i, P_j) = -\frac{1}{n^2(n+1)};\;\mbox{ for } i\neq j.
\]
Compared to the Bayesian bootstrap, the standard bootstrap utilizes the distribution
\begin{equation}\label{eq:fiducial-uniform-44}
  n\tilde{P} \sim \mbox{Multinomial}\left(n, \frac{1}{n}, ..., \frac{1}{n}\right),
\end{equation}
with
\[
\mbox{E}(\tilde{P}_i) = \frac{1}{n}, \quad 
\mbox{Var}(\tilde{P}_i) 
 = \frac{n-1}{n^3}, \quad\mbox{ and }\quad
 \mbox{Cov}(\tilde{P}_i, \tilde{P}_j) 
 =-\frac{1}{n^3};\;\mbox{ for } i\neq j.
\]
It becomes evident that the standard bootstrap (\ref{eq:fiducial-uniform-44}) can be seen as an approximation to the Bayesian bootstrap (\ref{eq:fiducial-uniform-42}) when $n$ is large, whereas the Bayesian bootstrap can be viewed as an  approximation to fiducial inference \eqref{eq:fiducial-bootstrap}. 
\end{example}
}

\subsection{The Generalized Association Method for Parametric Inference}
\label{ss:generalized-association}
Here, we provide a brief review of the valid inference discussed in \cite{martin2015plausibility}, which will be taken as our starting point in Section~\ref{s:par-inference}, where our proposed method will be introduced.
 Let $y$ be an observed sample of size $n$ from the true population $Y \sim \mathbf{P}_\theta$, $\theta\in\Theta$. Suppose that we are interested in inference on $\theta$.
We consider the problem of estimating $\theta$ with some loss function $\ell(y, \theta)$. Here, for our exposition, we focus on the likelihood-based loss function
\begin{equation}\label{eq:loss}
    \ell(y, \theta) = -\ln L_y(\theta) + \pi(\theta),\qquad \theta\in\Theta
\end{equation}
where $L_y(\theta)$ denotes the likelihood function and $\pi(\theta)$ stands for a penalty function. 
Following 
\citeauthor{martin2015plausibility} (\citeyear{martin2015plausibility}, see Equation (1)),
we take the \textit{generalized association function}
\begin{equation}\label{eq:T-association-01}
T_{y,\theta} =  \ell(y, \hat\theta_y) - \ell(y,\theta), \qquad\theta\in\Theta, y\in {\mathbb Y}
\end{equation}
where $\hat\theta_y = \arg\min_\theta\ell(y,\theta)$ and it is assumed that $\min_\theta\ell(y,\theta)$ is finite. 

As discussed in \cite{martin2015plausibility} and \cite{martin2015inferential}, valid inference can be made by making use of the distribution of $T_{Y,\theta}$,
\begin{equation}\label{eq:key-cdf}
F_\theta(t) = \mathrm{P}_{Y|\theta}\left(T_{Y,\theta}\le t\right),
\end{equation} 
where $Y \sim \mathbf{P}_\theta$ and $\theta \in \Theta$ is the fixed parameter value.
The basic idea behind inferential models is that based on the fact that given the value $U$ of $F_\theta\left(T_{y,\theta}\right)$, the knowledge we know about $\theta$ from the observed data $y$ is exactly represented by the set
\[
    \left\{\theta:\; F_\theta\left(T_{y,\theta}\right) = U
    \right\}.
\]
This leads to an exact frequentist confidence region $\{\theta:\;F_\theta(T_{y,\theta}) \geq \alpha\}$ with coverage probability $1-\alpha$ for all $\alpha \in (0, 1)$. More precisely, we have the following theoretical result \citep[{\it c.f.},][]{martin2015plausibility}.


\begin{theorem}[Exact confidence region] If $T_{Y,\theta}$ is a continuous random variable as a function of $Y \sim \mathbf{P}_\theta$, for any $\alpha \in (0, 1)$, the set $\{\theta:\;F_\theta(T_{Y,\theta}) \geq \alpha\}$ is an exact $1-\alpha$ frequentist confidence region. Namely, for $Y \sim \mathbf{P}_\theta$ and the fixed parameter value $\theta \in \Theta$, \label{thm:exact}
\[
\mathrm{P}_{Y|\theta} \left(\theta \notin \{\theta:\;F_\theta(T_{Y,\theta}) \geq \alpha\} \right) = \alpha
\]
\end{theorem}
\begin{proof}
If $T_{Y,\theta}$ is a continuous random variable as a function of $Y \sim \mathbf{P}_\theta$, the distribution of $F_\theta(T_{Y,\theta})$ follows the standard uniform distribution $\text{Uniform}(0,1)$. Therefore,
\begin{align*}
    \mathrm{P}_{Y|\theta} \left(\theta \notin \{\theta:\;F_\theta(T_{Y,\theta}) \geq \alpha\} \right) &= \mathrm{P}_{Y|\theta}\left(F_\theta(T_{Y,\theta}) \leq \alpha \right) \\
    &= \alpha,
\end{align*}
completing the proof.
\end{proof}

For the sake of clarity in the context of constructing confidence regions, we formally define the validity of a confidence-oriented inference procedure, which is consistent with that used in the bootstrap literature for considering large sample-based validity.
\begin{definition}[Valid Parametric Inference] A confidence-oriented inference procedure for the model parameter $\theta \in \Theta$ is said to be valid at level $\alpha\in (0,1)$ if its $1-\alpha$ confidence region $\mathbf{C}(Y)$ satisfies  
\label{def:valid}
\[
\mathrm{P}_{Y|\theta}\left(\theta \notin \mathbf{C}(Y) \right) \leq \alpha
\]
as $Y \sim \mathbf{P}_\theta$. It is said to be valid if it is valid at all levels.
\end{definition}


Thus, from Definition \ref{def:valid} and Theorem~\ref{thm:exact}, we see that the key to obtaining the valid inference based on $T_{y,\theta}$ is to be able to evaluate $F_\theta(t)$ defined in \eqref{eq:key-cdf}. Here is a simple example where $F_\theta(t)$ can be obtained analytically.

\begin{remark}
  The pioneering work of \cite{martin2015plausibility}, developed in the framework of IMs, can be viewed as to use the set $\{\theta:\;F_\theta(T_{y,\theta}) \geq \alpha\}$ to construct an exact confidence region for $\theta.$  However, it is not always feasible to derive a closed-form expression of $F_\theta(T_{y,\theta})$, necessitating Monte Carlo estimation. Meanwhile, evaluating the function value across a wide range of $\theta$ is required. Consequently, as is pointed out by \cite{martin2015plausibility}, the computational cost can become prohibitively high, especially when dealing with high-dimensional $\theta$.   
\end{remark}

\begin{remark}
    Besides \eqref{eq:T-association-01}, the generalized association function $T_{y,\theta}$ can take other forms \citep{martin2015plausibility}. Although the algorithm proposed in Section~\ref{s:SA} does not depend on this specific formulation, we choose to adhere to \eqref{eq:T-association-01} throughout our discussion. The choice is motivated by the fact that (\ref{eq:T-association-01}) conveniently summarizes all information in $y$ concerning $\theta$, and (\ref{eq:T-association-01}) is computationally simple compared to other potential forms of $T_{y,\theta}$, such as the standardized signed log likelihood ratio \citep{Barndorff1986}.
\end{remark}

\subsection{$T$-Confidence Distribution}\label{ss:confidence-distribution}
Utilizing the valid confidence set $\{\theta:\;F_\theta(T_{y,\theta}) \geq \alpha\}$ discussed in Section~\ref{ss:generalized-association}, suppose now we are interested in constructing a distribution of a function of both the unknown parameter $\theta$ and the data $y$, where $y$ is taken as random variable, 
such that it assigns at least probability $1-\alpha$ to all the $100(1-\alpha)\%$ confidence sets. Intuitively, for each fixed $y$, our goal is to create a distribution of the parameter $\theta$, denoted by $G_y$, such that for samples drawn from $G_y$ and sorted according to the value $F_\theta(T_{y,\theta})$, the proportion of these values greater than $\alpha$ is at least $1-\alpha$ for any $\alpha \in (0,1)$. By Theorem~\ref{thm:exact}, this guarantees the desired property of the distribution $G_y$.  Such a distribution will greatly facilitate the inference procedure, as the choice of $\alpha$ and the confidence set is arbitrary. Formally, we define such a distribution as a \textit{T-confidence distribution} as follows \citep[{\it c.f.},][]{martin2021inferential, martin2023b}:
\begin{definition}[$T$-Confidence Distribution] Given the observed data $y$ from $Y\sim \mathbf{P}_\theta$ and a function $T_{y, \theta}$ of $y$ and $\theta$, a distribution $G_y$, indexed by $y$ and on the space of $\theta$, is said to be a $T$-confidence distribution of the unknown parameter $\theta$ with respect to $T_{y, \theta}$ if  $G_y$ satisfies:
\label{def:confidence}
\begin{equation} \label{eq:confidence}
    \mathrm{P}_{\theta^*}\left(F_{\theta_*}(T_{y, \theta_*})\le\alpha\right) = \alpha,\quad \text{for any  $\alpha \in (0,1)$,}
\end{equation}
or equivalently,
\(F_{\theta_*}(T_{y,\theta_*}) \sim \text{Uniform}(0,1)\), when $\theta_* \sim G_y$.
\end{definition}

For simplicity, hereafter we refer to the $T$-confidence distribution defined above simply as the \textit{confidence distribution}.
It follows that confidence sets can be produced straightforwardly from confidence distributions, as shown in the following theorem.

\begin{theorem} Suppose that $G_y$ is a confidence distribution with respect to $T_{y, \theta}$ for given the observed data $y$ from $Y \sim \mathbf{P}_\theta$.
For any $\alpha \in (0,1)$, let $U^{1-\alpha}_y$ be the subset of $\Theta$ determined by
$\int_{U^{1-\alpha}_y}dG_y \geq 1-\alpha$.
Then $U^{1-\alpha}_y$, as a confidence set for the true parameter $\theta$, has at least $100(1-\alpha)\%$ coverage probability, that is, for $Y \sim \mathbf{P}_\theta$,
\[
\mathrm{P}_{Y|\theta}\left(\theta \notin U^{1-\alpha}_Y \right) \leq \alpha.\]\label{thm:2}
\end{theorem}
\begin{proof}
The function mapping $\theta \mapsto F_\theta(T_{y,\theta})$ gives $U^{1-\alpha}_y \mapsto S^{1-\alpha}_y \subset (0,1)$, with $\mathrm{P}_{\theta^*}\left(F_{\theta_*}(T_{y, \theta_*}) \in S^{1-\alpha}_y \right) \geq 1-\alpha$. Since \(F_{\theta_*}(T_{y,\theta_*})\) follows the standard uniform distribution denoted by $U$, we have $\int_{S^{1-\alpha}_y}dU \geq 1-\alpha$. It follows from Theorem~\ref{thm:exact} that, when $Y \sim \mathbf{P}_\theta$,
\[
    \mathrm{P}_{Y|\theta} \left(\theta \notin \{\theta:\;F_\theta(T_{Y,\theta}) \in S^{1-\alpha}_Y\} \right) = \mathrm{P}_{Y|\theta}\left(F_\theta(T_{Y,\theta}) \notin S^{1-\alpha}_Y \right) 
    \leq \alpha,
\]
completing the proof.
\end{proof}

\begin{remark}
    A more rigorous definition of the confidence distribution from the imprecise probability point of view \citep{martin2021imprecise, martin2023b} is given in Supplementary~S.1. Furthermore, \cite{martin2023b} shows that if well-defined, the fiducial distribution obtained via the traditional fiducial inference \cite[see, e.g.,][and references therein for more discussions on fiducial]{Zabell1992,liu2015frameworks} is precisely the confidence distribution. 
\end{remark}




\begin{example}[Inference on the unknown mean of a normal distribution]\label{example:f-normal}
\rm
Suppose that we have observed a sample of size $n$, $y := (y_1,..., y_n)$, from the normal distribution $N(\theta, 1)$ with unknown $\theta \in {\mathbb R}$. We are interested in making valid inference on the mean parameter $\theta.$ Utilizing $T_{y,\theta}$ and its distribution, we have
\begin{equation}\label{eq:simple-T}
T_{y,\theta} = 
 \ell(y, \hat\theta) - \ell(y,\theta)  = \frac{1}{2}\left[ \sum_{i=1}^n(y_i-\hat\theta)^2 - \sum_{i=1}^n(y_i-\theta)^2 \right]
= -\frac{n}{2}(\bar{y}-\theta)^2,
\end{equation}
where $\hat\theta=\bar{y}$, the sample mean $\bar{y} = \frac{1}{n}\sum_{i=1}^n y_i$. This implies that $T_{Y,\theta}\sim -\frac{1}{2}\chi_1^2$, as the sampling distribution of $\bar{y}$ is $N(\theta, 1/n)$. Represent the standard normal distribution as $Z \sim N(0,1)$ and symbolize its cumulative distribution function as $\Phi(\cdot)$. For $Y \sim \mathbf{P}_\theta$,
\begin{align*}
F_\theta(t) = \mathrm{P}_{Y|\theta}\left(T_{Y,\theta} \le t\right) &= \mathrm{P}\left(\chi_1^2\geq -2 t\right) \\
&= 2\Phi(-\sqrt{-2 t}),\qquad t\leq 0.  
\end{align*}
Note that $\{\theta:\;F_\theta(T_{y,\theta}) \geq \alpha\}$ gives exactly the same classical $(1- \alpha)\%$ $z$-interval for $\theta$.

By applying the reverse operation, we obtain the distribution $\theta_* \sim N (\bar{y}, 1/n)$. We have $F_{\theta_*}(T_{y,\theta_*}) = 2 \Phi (-\sqrt{n} |\bar{y} - \theta_*|) \sim 2 \Phi (-|Z|)$ where $Z \sim N(0,1)$. Since $2 \Phi (-|Z|) \sim \text{Uniform}(0,1)$, the distribution of $\theta_*$ satisfies (\ref{eq:confidence}), and thus is the confidence distribution of $\theta$. In other words, any $100(1-\alpha)\%$ confidence set from $N (\bar{y}, 1/n)$ has at least $1-\alpha$ chance to cover the true parameter $\theta$.
\end{example}

\section{The Calibrated Bootstrap Method}
\label{s:SA}
In this section, we propose a computationally feasible method to perform valid parametric inference via an adaptive resampling procedure.  This method, termed {\it calibrated bootstrap} (CB), consists of two steps: \textit{resampling approximation} (RA) and \textit{distributional resampling} (DR). The RA step searches for a resampling scheme that can best approximate the exact confidence region of $\theta\in\Theta$ for each of a set of prespecified confidence coefficients. The DR step selects samples from the bootstrapped estimates obtained in the RA step to construct the confidence distribution of the unknown parameter of interests. These two steps are summarized as two algorithms and are discussed below in Sections \ref{ss:SA} and \ref{ss:SA-fiducial}.

\label{s:par-inference}\label{ss:parametric-inference}
\subsection{Resampling Approximation}
\label{ss:SA}
It is seen in Section~\ref{ss:confidence-distribution} that the key to producing valid inference based on $T_{y,\theta}$ is to obtain the confidence distribution $\theta_* \sim G_y$ that satifies (\ref{eq:confidence}). For the general case,  in the CB context, we are interested in a bootstrapped sample (or, in theory, population) of estimates $\tilde{\Theta}$, represented in terms of the empirical distribution $\tilde{G}_y$ based on $\tilde{\Theta}$ and referred to as {\it bootstrapped confidence distribution} in the sequel, as long as it provides valid inference in the sense:
\begin{equation}\label{eq:fiducial-dist}
\hat{\theta}_* \sim \tilde{G}_y \ \ s.t. \ \  F_{\hat{\theta}_*}(T_{y,\hat{\theta}_*}) \sim \text{Uniform}(0,1).
\end{equation}
It appears difficult to design a resampling scheme to construct a $\tilde{\Theta}$ satisfying \eqref{eq:fiducial-dist}. Here, we consider to find $\tilde{\Theta}_\alpha$ with the corresponding empirical distribution $\tilde{G}_{y}^{\alpha}$ such that for random variable $\hat{\theta}_*$ and distribution function $\tilde{G}_{y}^{\alpha}$ that depends on given $\alpha$ and $y$, when $\hat{\theta}_* \sim \tilde{G}_{y}^{\alpha}$,
\begin{equation}\label{eq:fiducial-alternative-obj}
    \mathrm{P}_{\hat{\theta}_*}\left(F_{\hat{\theta}_*}(T_{y, \hat{\theta}_*})\le\alpha\right) = \alpha, \quad \text{for some pre-defined $\alpha \in (0,1)$.}
\end{equation}
This is an easier alternative to finding $\tilde{\Theta}$ as \eqref{eq:fiducial-dist} entails \eqref{eq:fiducial-alternative-obj}. The following Lemma suggests that when $\hat{\theta}_* \sim \tilde{G}_{y}^{\alpha}$, the obtained $1-\alpha$ confidence region using the 
$\alpha$ quantile of $F_{\hat{\theta}_*}(T_{y, \hat{\theta}_*})$ 
is exact.
\begin{lemma}\label{lemma:1} Suppose $T_{Y,\theta}$ is a continuous random variable as a function of $Y \sim \mathbf{P}_\theta$ for $\theta \in \Theta$. For a distribution $\tilde{G}_{y}^{\alpha}$ such that when $\hat{\theta}_* \sim \tilde{G}_{y}^{\alpha}$, (\ref{eq:fiducial-alternative-obj}) holds for some pre-defined  $\alpha \in (0,1)$. Denote the distribution of $F_{\hat{\theta}_*}(T_{y, \hat{\theta}_*})$ as $Q_{y}$. Suppose that $\tilde{G}_{y}^{\alpha}$ has non-zero density at any $\theta_0 \in \{\theta:\;F_\theta(T_{y,\theta}) \geq \alpha\}$ and zero density at any $\theta_0 \notin \Theta$. Then, the set
\begin{equation}\label{eq:CI-region}
    \{\theta':F_{\theta'}(T_{y, \theta'}) \geq Q^{-1}_{y}(\alpha)\}, 
\end{equation}
where each $\theta'$ represents a realization of the random variable $\hat{\theta}_*$ simulated from $\tilde{G}_{y}^{\alpha}$,
provides an exact $1-\alpha$ confidence region for $\theta$ as defined in Theorem~\ref{thm:exact}, where $Q^{-1}_{y}(\cdot)$ denotes the inverse function of $Q_{y}$.
\end{lemma}
\begin{proof}
 Condition (\ref{eq:fiducial-alternative-obj}) implies that $Q^{-1}_{y}(\alpha) = \alpha$. Since  $\tilde{G}_{y}^{\alpha}$ has non-zero density at any $\theta_0 \in \{\theta:\;F_\theta(T_{y,\theta}) \geq \alpha\}$, (\ref{eq:CI-region}) is equivalent to $\{\theta:\;F_\theta(T_{y,\theta}) \geq \alpha\}$. The exactness follows from Theorem~\ref{thm:exact}.
\end{proof}

\ifthenelse{\boolean{inline}}{
\begin{figure}
    \centering
    \begin{minipage}{0.45\linewidth}
        \centering
        \includegraphics[width=6cm]{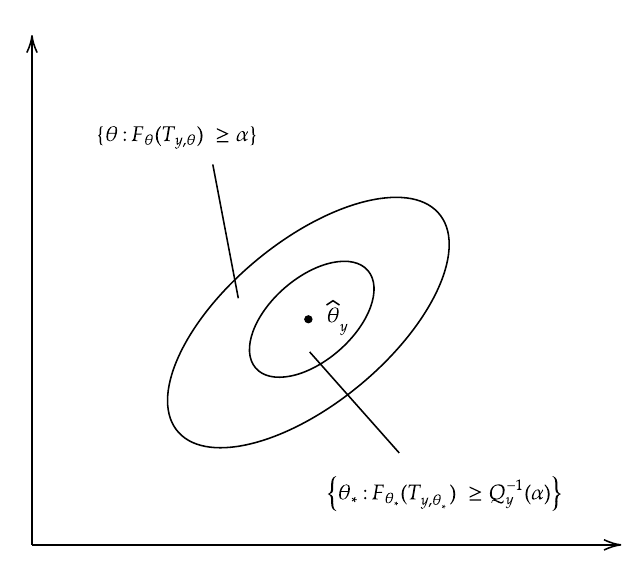}
        \textbf{(a)} $m > m_\alpha$
    \end{minipage}
    \hfill 
    \begin{minipage}{0.45\linewidth}
        \centering
        \includegraphics[width=6cm]{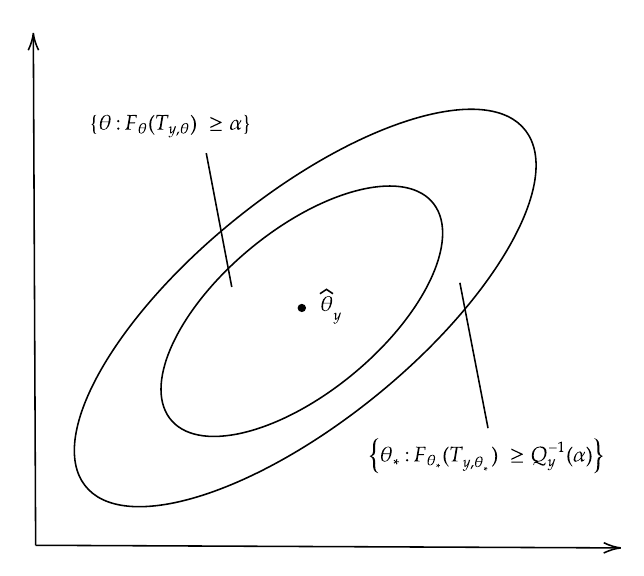}
        \textbf{(b)} $m < m_\alpha$
    \end{minipage}
    \caption{Illustration of the over-disperseness of $\hat{\theta}_*$ with different choice of $m$. The distribution of $F_{\hat{\theta}_*}(T_{y, \hat{\theta}_*})$ for $\hat{\theta}_*$ obtained by the $m$-out-of-$n$ bootstrap is denoted by $Q_{y}$, and its inverse function is denoted by $Q^{-1}_{y}(\cdot)$. The optimal $m$ denoted by $m_\alpha$ is the $m$ such that the confidence region $\{\hat{\theta}_*: F_{\hat{\theta}_*}(T_{y,\hat{\theta}_*}) \geq Q^{-1}_{y}(\alpha) \}$ closely matches the desired confidence region $\{\theta: F_{\theta}(T_{y,\theta}) \geq \alpha \}$.}
    \label{fig:contour}
\end{figure}
}{}

Denote by $r(\cdot)$ a resampling scheme $r(\cdot)$ that is used to generate the resampled data set $\tilde{y} \sim r(y)$. The variability of the estimates introduced by $r(\cdot)$ is captured through the distribution of bootstrap estimates, denoted by $\hat{\theta}_*$. Here, $\hat{\theta}_*$ is defined as the estimated parameter that minimizes the loss function for a given resampled dataset $\tilde{y}$:
\begin{equation}\label{eq:resampling-1}
    \hat{\theta}_* = \argmin_\theta\ell(\tilde{y},\theta), \quad \tilde{y} \sim r(y)
\end{equation}
Ideally, for valid inference, the variability of the estimates of $\theta$ in (\ref{eq:resampling-1}) should match its theoretically quantified uncertainty in (\ref{eq:fiducial-dist}) or (\ref{eq:fiducial-alternative-obj}). 
To search for such a resampling scheme, here we consider a more general method of $m$-out-of-$n$ bootstrap \citep{bickel2008choice}. That is,
\begin{equation}\label{eq:m-out-of-n}
    r(y) = \{\tilde{y}_1, \dots, \tilde{y}_m\}, \qquad m \geq 1
\end{equation}
with $\tilde{y}_i \ (i = 1,\dots,m)$ being a random draw with replacement from $\{y_1,..., y_n\}$. Note that (\ref{eq:m-out-of-n}) reduces to the standard bootstrap when $m = n$. We take $m$ here to be adaptive to specific models and observations in order to control the variability of the estimates of $\theta$ in (\ref{eq:resampling-1}). Numerically, as $m \to \infty$, we have $\hat{\theta}_* \stackrel{p}{\rightarrow} \hat{\theta}_y$, and $F_{\hat{\theta}_*}(T_{y,\hat{\theta}_*}) \stackrel{p}{\rightarrow} 1$, where $\stackrel{p}{\rightarrow}$ denotes convergence in probability. Theoretically, as a matter of fact, an asymptotic normality result similar to that for the maximum likelihood estimator \citep[see, {\it e.g.},][p. 415]{lehmann1991theory} can be established by treating the $\hat\theta=\argmin_{\theta\in\Theta}\ell(y,\theta)$ as the true parameter with respect to resampling, because the same proof for the maximum likelihood estimator can go through here due to the fact that the expectation of the score function at $\theta = \hat\theta$ is zero.
Conversely, as $m$ decreases, $\hat{\theta}_*$ diverges from $\hat{\theta}_y$, causing the distribution of $F_{\hat{\theta}_*}(T_{y,\hat{\theta}_*})$ to gradually shift toward 0. The effect of $m$ on the over-disperseness of the estimates $\hat{\theta}_*$ is illustrated in Figure~\ref{fig:contour}. 

For theoretical justifications, here we make a mild assumption that the inference problem is \textit{bootstrap $\epsilon$-calibratable} at level $\alpha$. Specifically, for fixed $y$, there exists a positive integer $m$ such that the $m$-out-of-$n$ bootstrapped estimates $\hat{\theta}_*$ satisfy
\begin{equation}\label{eq:bootstrap=calibratable}
        \alpha \leq \mathrm{P}_{\hat{\theta}_*}\left(F_{\hat{\theta}_*}(T_{y, \hat{\theta}_*})\leq \alpha\right) \leq \alpha + \epsilon,
\end{equation}
for some $\epsilon \in [0,\alpha]$ is the tolerance. This assumption ensures the existence of an approximate solution to (\ref{eq:fiducial-alternative-obj}). Additional insights of such a property are given in Supplementary~S.2.3. The assumption is considered mild for some common choices of $\alpha$ such as $\{0.05,0.10,\dots,0.95\}$ due to the well-documented flexibility and efficiency of the adaptive $m$-out-of-$n$ bootstrap in the literature \citep{bickel2008choice, Chakraborty2013, jiang2024}.
Similarly, we also make another mild assumption that the set of the finite number of possible values provided by the $m$-out-of-$n$ bootstrap is dense enough for practically accurate inference on $\theta$ given $y$; see the discussion in Section \ref{s:discussion} on the use of weighted likelihood estimation as alternatives to or a generalization of resampling in the context of bootstrapping.

In the first step of the proposed CB method,  our primary objective is to determine an optimal value of $m$, such that the empirical distribution of the obtained bootstrap estimates of $\theta$ in (\ref{eq:resampling-1}) from the $m$-out-of-$n$ bootstrap best approximates the distribution of $\tilde{G}_{y}^{\alpha}$ in (\ref{eq:fiducial-alternative-obj}).  This can be done using the {\it stochastic approximation} (SA) algorithm \citep{robbins1951stochastic}. See \cite{Syring2019} for the application of SA for obtaining a specific target coverage level in a different context.
For an unbiased estimator of the distribution of $F_{\hat{\theta}_*}(T_{y, \hat{\theta}_*})$ that is needed for a SA implementation,
we consider estimating the distribution of $F_{\hat{\theta}_*}(T_{y, \hat{\theta}_*})$ with the empirical distribution function obtained by $B$ resampling repetitions. Suppose that after $B$ resampling repetitions, we have $B$ simulated values
\begin{eqnarray*}
    U_1 &=&F_{\hat{\theta}_*^{(1)}}(T_{y, \hat{\theta}_*^{(1)}})\\
    &\vdots&\\
    U_B &=&F_{\hat{\theta}_*^{(B)}}(T_{y, \hat{\theta}_*^{(B)}}),
\end{eqnarray*}
where $\hat{\theta}_*^{(b)} \ (b = 1,\dots,B)$ denotes the obtained estimation of $\theta$ by minimizing the loss functions with regard to the $b$-th bootstrap samples. For sufficiently large $B$, objective (\ref{eq:fiducial-alternative-obj}) is equivalent to
\begin{equation}\label{eq:fiducial-constraint-03}
   \lim_{B\rightarrow\infty} \frac{1}{B}\sum_{b = 1}^{B}\mathbb{I}(U_b \leq \alpha) = \alpha
\end{equation}
where $\mathbb{I}(\cdot)$ denotes the indicator function. The approximation of $F_{\hat{\theta}_*}(T_{y, \hat{\theta}_*})$ for each $\hat{\theta}_*$ can be done by sampling $N$ new observation vectors $y^{(1)}, \dots, y^{(N)}$ from the population $\mathbf{P}_{\hat{\theta}_*}$
\begin{equation}\label{eq:fiducial-constraint-04}
  F_{\hat{\theta}_*}(T_{y, \hat{\theta}_*}) \approx \frac{1}{N}\sum_{i = 1}^{N}\mathbb{I}\left( T_{y^{(i)}, \hat{\theta}_*} \leq T_{y, \hat{\theta}_*} \right).  
\end{equation}
Alternative approaches, such as by reweighting the observed data using the SIR approach of \cite{rubin1987calculation}, can also be considered.


\ifthenelse{\boolean{inline}}{
\begin{algorithm}
        Specify the target coverage level $\alpha$ and 
        choose the number of bootstrap replications $B$\;
        Initialize $m^{(0)}$ with $m^{(0)}=n$ and the iteration number $t$ with $t = 0$\;
        \While{$m^{(t)}$ not converged} {
            Let $m=\lfloor m^{(t)}\rfloor + \mbox{Bernoulli}(m^{(t)} -\lfloor m^{(t)}\rfloor)$, where $\lfloor m^{(t)}\rfloor$ denotes the largest integer not greater than $m^{(t)}$ and $\mbox{Bernoulli}(p)$ a Bernoulli random variable with parameter $p$\;
            Sample  $\{(x_i^*, y_i^*):\; i=1,..., m^{(t)}\}$ or, in the matrix notation, $(X^*, Y^*)$ from $\{(x_i, y_i):\; i=1,..., n\}$ with replacement\;    
            Compute $\hat{\theta}_* = \argmin_\theta \ell(\theta, X^*, Y^*)$\;
            Evaluate $\ell^{(t)} = -[\ell(\hat{\theta}_*, X, Y) - 
            \ell(\hat{\theta}, X, Y)]$\;
            Initialize $P=0$\;
            \For{$b \gets 1$ \KwTo $B$}{
                Sample $\{(x_i^{**}, y_i^{**}):\; i=1,..., n\}$ or, in the matrix notation, $(X^{**}, Y^{**})$ from model $\mathbf{P}_{\hat{\theta}_*}$\;
                Compute $\hat{\theta}_{**} = \argmin_\theta \ell(\theta, X^{**}, Y^{**})$\;
                Evaluate $S^{(b)} = -[\ell(\hat{\theta}_{*}, X^{**}, Y^{**}) - \ell(\hat{\theta}_{**}, X^{**}, Y^{**})]$\;
                \If{$S^{(b)} \le \ell^{(t)}$}{
                    $P\gets P + 1$\;
                }
            }
            Set $Z_t \gets \frac{\mathbb{I}{\{P/B \le\alpha\}}  - \alpha}{\sqrt{B \alpha(1-\alpha)}}$\;
            Update $m^{(t+1)} = m^{(t)} + \frac{c}{t + 1} Z_t$, where $c$ is a predefined constant\;
            Set $t\gets t+1$\;
        }
        Return $\lfloor m^{(t)}\rfloor - 1$.
   \caption{The Resampling Approximation algorithm\label{alg:SA4LR-parameters} 
   }
\end{algorithm}
}{}

The right-hand side of \eqref{eq:fiducial-constraint-04} is seen as an unbiased estimator 
of $F_{\hat{\theta}_*}(T_{y, \hat{\theta}_*})$.
Making use of this fact, we propose a computationally efficient SA method of finding the optimal $m$ for a pre-specified value of $\alpha \in(0,1)$ subject to constraints \eqref{eq:fiducial-constraint-03}. The method is summarized as Algorithm \ref{alg:SA4LR-parameters}. The step size $c$ used in the algorithm controls the convergence speed, and one practical choice is $c = d\cdot n$ for some $d > 0$. In the simulation study in Section~\ref{s:application}, we use $d = 10$.  Additionally, in the illustrative examples presented in this paper, we use a small number of resampling repetitions, specifically $B = 10$. This approach has yielded satisfactory convergence results. The following theorem provides necessary theoretical results on the required stochastic update in Algorithm \ref{alg:SA4LR-parameters}. The proof is given in the Supplementary~S.2.

\begin{theorem}\label{Thm:convergence-1}
For a given targeted coverage probability $1-\alpha \ (0 < \alpha < 1)$, let the observed coverage probability be defined as a function of $m$:
\[
f^\alpha(m) \coloneqq \mathrm{P}_{\hat{\theta}_*(m)}\left(F_{\hat{\theta}_*(m)}(T_{y, \hat{\theta}_*(m)})\le\alpha\right)
\]
where $\hat{\theta}_*(m)$, with the resulting distribution $\mathrm{G}_{y,m},$ 
is the bootstrap estimate of $\theta$ obtained with the $m$-out-of-$n$ bootstrap. 
The convergence of Algorithm~\ref{alg:SA4LR-parameters} to the desired solution
\[ m_\alpha^* = \min_{m \geq 1} \{ f^\alpha(m) -  \alpha : f^\alpha(m) \geq \alpha \} \]
 is guaranteed with probability one if $f^\alpha(m)$ is non-increasing with respect to $m$ and
\begin{equation}\label{theom.eq.1}
|f^\alpha(m_0) - \alpha| > |f^\alpha(m_\alpha^*) - \alpha| 
\end{equation}
for all $m_0 \neq m_\alpha^*$, under the regularity conditions of likelihood function given in Supplementary~S.2. Moreover, if the problem is bootstrap $\epsilon$-calibratable by (\ref{eq:bootstrap=calibratable}), the constructed confidence interval by (\ref{eq:CI-region}) is at least at the $1-\alpha$ level.

\color{black}





\end{theorem}

\begin{corollary}\label{Cor:convergence-1}
The sufficient conditions of non-increasing $f^\alpha(m)$ in Theorem~\ref{Thm:convergence-1} are:
\begin{enumerate}[{C}1.]
    \item For $\theta_1$, $\theta_2$ such that $\ell(y, \theta_1) \geq \ell(y, \theta_2)$, we have $F_{\theta_1}(T_{y, \theta_1}) \leq F_{\theta_2}(T_{y, \theta_2})$, and
    \item For any $m_1 > m_2 > 0$,  $\ell(y, \hat{\theta}_*(m_2))$ first-order stochastically dominates $\ell(y, \hat{\theta}_{*}(m_1))$.
\end{enumerate}
Moreover, Theorem~\ref{Thm:convergence-1} holds asymptotically as $m,n \to \infty$, given that the asymptotic normality of maximum likelihood estimator (MLE) holds.
\end{corollary}
\begin{remark}
    One sufficient condition for C1 in Corollary~\ref{Cor:convergence-1} to hold true is the distribution of $T_{Y,\theta}$ being independent of $\theta$ for $Y \sim \mathrm{P}_\theta$. Further insights about this sufficient condition are elaborated in \cite{martin2015plausibility} as well as Supplementary~S.3. It can also be verified that condition C1 is satisfied in commonly used models, such as the linear regression model. Condition C2, on the other hand, is grounded in the concept that the model loss decreases as the sample size increases. The demonstration of the example of linear regression model satisfying these conditions are given in Supplementary~S.3.2. 
\end{remark}

Now, we attempt to approximate the confidence distribution $\hat{\theta}_* \sim \tilde{G}$ in (\ref{eq:fiducial-dist}) with some resampling scheme. Note that the sufficient condition for (\ref{eq:fiducial-dist}) is that for all $\alpha \in (0,1)$, (\ref{eq:fiducial-alternative-obj}) holds true. This motivates us to use the SA method to find the optimal $m$ for each in a range of target coverage probabilities, such as $0.05, 0.15, ..., 0.95$. The estimated resampling scheme is the $m$-out-of-$n$ bootstrap with a mixture of these $m$ values. Note that due to the limitation in the versatility of $m$-out-of-$n$ bootstrap, the confidence distribution may not be perfectly recovered. In this case, further refinement methods, as will be dicussed in Section~\ref{ss:SA-fiducial}, are necessary.

\subsection{RA with a Simple Example} \label{ss:ra-simple}
Consider the example in Section~\ref{ss:confidence-distribution} on the inference of the mean parameter $\theta$ with $y$, a sample containing $n$ observations from the model $Y \sim N(0,1)$. Here we illustrate the application of RA to numerically approximate the $(1- \alpha)\%$ $z$-interval for $\theta$ with some pre-specified $\alpha$.

By (\ref{eq:simple-T}), for each bootstrap estimate $\hat{\theta}^*\in\Theta$, the Monte-Carlo estimated function value $F_{\hat{\theta}_*}(T_{y, \hat{\theta}_*})$ with $B$ repetitions takes the form 

\begin{align}
\begin{split}
  F_{\hat{\theta}_*}(T_{y, \hat{\theta}_*}) &\approx \frac{1}{B}\sum_{b = 1}^{B}\mathbb{I}\left( T_{y^{(b)}, \hat{\theta}_*} \leq T_{y, \hat{\theta}_*} \right) \\
  &= \frac{1}{B}\sum_{b = 1}^{B}\mathbb{I}\left( \frac{1}{2}\sum_{i=1}^n (\bar{y}^{(b)} - y_i^{(b)})^2 - \frac{1}{2}\sum_{i=1}^n (\hat{\theta}_* - y_i^{(b)})^2 \leq \frac{1}{2}\sum_{i=1}^n (\bar{y} - y_i)^2 - \frac{1}{2}\sum_{i=1}^n (\hat{\theta}_* - y_i)^2 \right)  
\end{split}\label{eq-simple-1}
\end{align}
where $y^{(b)}$ is a sample of size $n$ from the model $N(\hat{\theta}_*, 1)$. 

In the RA process, an initial value $m^{(0)}$ is used to conduct the $m$-out-of-$n$ bootstrap on the observed data $y$ to obtain the resample data $\tilde{y}$, which yields the parameter value $\hat{\theta}_* = \bar{\tilde{y}}$,  the MLE of the resampled data. Next, the function value $F_{\hat{\theta}_*}(T_{y, \hat{\theta}_*})$ is calculated by (\ref{eq-simple-1}), based on which we update $m^{(0)}$ with
\[
m^{(t+1)} = m^{(t)} + \frac{c}{t + 1}\frac{\mathbb{I}{\{F_{\hat{\theta}_*}(T_{y, \hat{\theta}_*}) \le\alpha\}}  - \alpha}{\sqrt{B \alpha(1-\alpha)}},
\]
taking the iteration number as $t = 0$, where $c = 10 \cdot n$ is a predefined constant. Finally, we increase the iteration number $t$ by 1 and repeat the whole process for the repetitive updates of $m$ until convergence.

With the obtained $m$ from the RA step, we can then conduct the $m$-out-of-$n$ bootstrap to obtain the resampled $\hat{\theta}_*$ whose distribution satisfies
\(
\mathrm{P}_{\hat{\theta}_*}\left(F_{\hat{\theta}_*}(T_{y, \hat{\theta}_*})\le\alpha\right) \approx \alpha.
\)
By Theorem~\ref{thm:exact}, we wish to construct the $1-\alpha$ confidence interval
\(
    \{\hat{\theta}_*:F_{\hat{\theta}_*}(T_{y, \hat{\theta}_*}) \geq \alpha\}
\)
with the set of $\hat{\theta}_*$ values obtained from the $m$-out-of-$n$ bootstrap.
It can be seen from (\ref{eq-simple-1}) that the distribution of the left hand side of the inequality does not depend on $\hat{\theta}_*$. As a result, for $\theta_1$, $\theta_2$ such that $\ell(y, \theta_1) \geq \ell(y, \theta_2)$, we have $F_{\theta_1}(T_{y, \theta_1}) \leq F_{\theta_2}(T_{y, \theta_2})$. Thus, the $1-\alpha$ confidence interval can be obtained by the range of the $\hat{\theta}_*$ values for which the corresponding loss function values fall within the lowest $1-\alpha$ quantile of all obtained loss function values. 

For numerical evaluations, we experimented with the case $n = 50, \theta = 1, \alpha = 0.05$. The theoretical interval is $(0.757, 1.312)$ and the 
proposed RA method gives $(0.756,1.313)$.

\subsection{RA with a Linear Regression Example}
To illustrate the proposed method,  here we present a study on linear regression under a high-dimensional setting.
Further applications of Algorithm \ref{alg:SA4LR-parameters} are illustrated in Section~\ref{s:application}.

\begin{example}[High-Dimensional Linear Regression]\label{example:lr}
\rm 
Consider the linear regression model
\begin{equation}\label{eq:LR01}
    Y = X\beta + \varepsilon,\quad \varepsilon \sim N_n(0,\sigma^2 I), \beta \in {\mathbb R}^p,
\end{equation}
where $X$ is the $(n\times p)$ full-rank design matrix,  $Y$ is the vector of $n$ observed responses, $1\leq p<n$, and $\sigma^2 > 0$ denotes the variance of the noise term. In the context of high-dimensional linear regression, one may imagine the scenario where $p/n \to c$ for some $0 < c < 1$ as $n \to \infty$.

Suppose that the estimand of interest is the unknown vector $\beta$ of regression coefficients. Theoretical results from the fiducial inference perspective or, more precisely, the CIM (Conditional Inferential Models) perspective \citep[see, {\it e.g.},][]{martin2015inferential} shows that for the estimates
\[
\hat{\beta} = (X'X)^{-1}X' Y, \quad \hat{\sigma} =  \frac{1}{n-p}Y'\left(I- X(X'X)^{-1}X'\right)Y,
\]
the fiducial or confidence distribution of $\beta$ is given by
\begin{equation}\label{eq:fiducial-theoretical-1}
    \beta^* \sim t_p\left(\hat{\beta}, \hat{\sigma}^2 (X'X)^{-1}, n-p\right),
\end{equation}
where $t_p(\cdot)$ denotes the $p$-dimensional multivariate student-$t$ distribution, and if $\sigma$ is known,
\begin{equation}\label{eq:fiducial-theoretical-2}
    \beta^* \sim N_p\left(\hat{\beta}, \sigma (X'X)^{-1}\right)
\end{equation}
where $N_p(\cdot)$ denotes the $p$-dimensional multivariate Gaussian distribution. Apparently, (\ref{eq:fiducial-theoretical-1}) and (\ref{eq:fiducial-theoretical-2}) can be used to construct the classical valid frequentist confidence region for $\beta$. More relevant details and discussion, focusing on the high-dimensional setting, are provided in Supplementary~S.5 and S.6.

\ifthenelse{\boolean{inline}}{
\begin{figure}[!htp]
\centering
\includegraphics[width=5.5in]{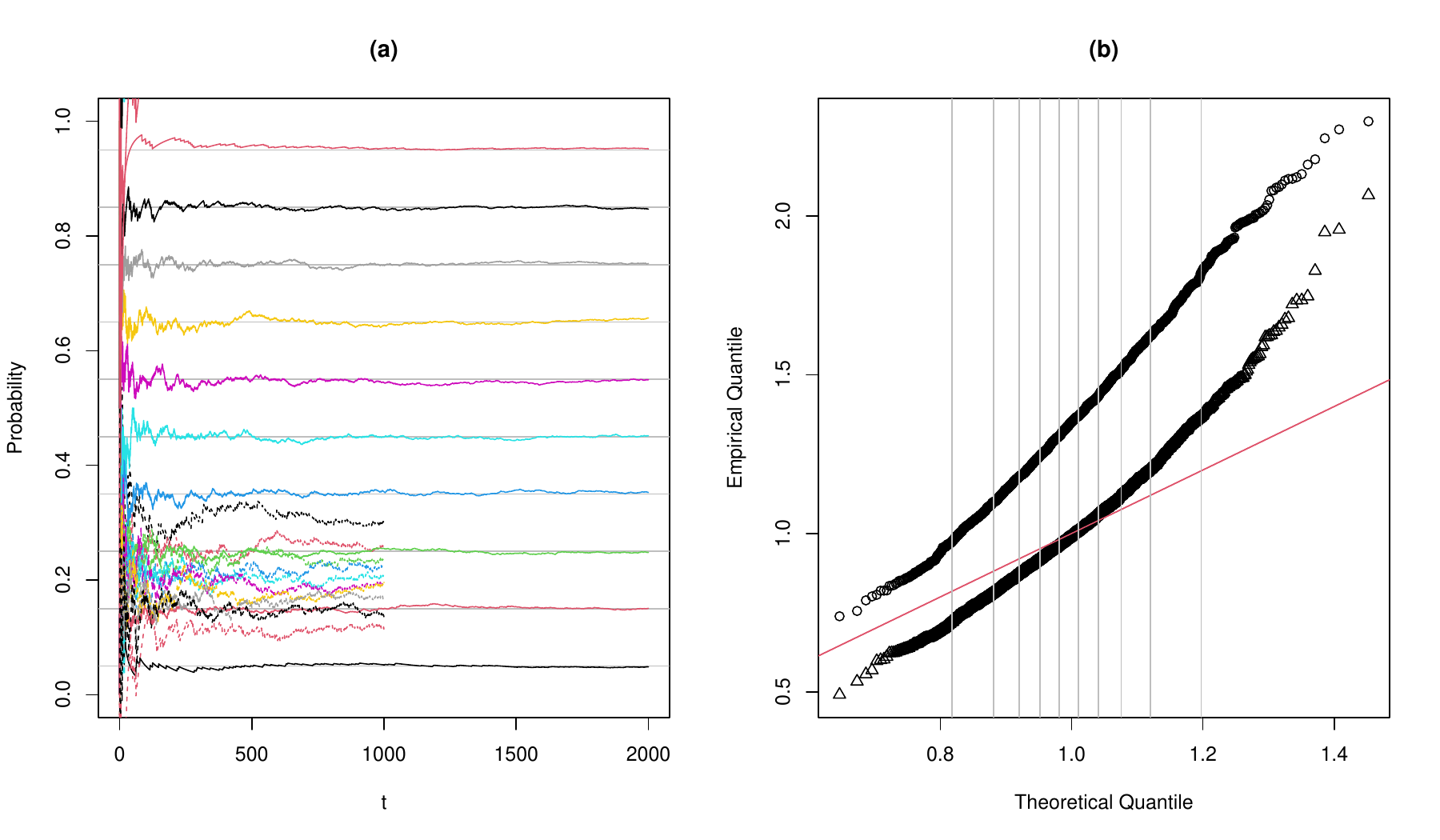}
        \caption{Illustration of SA on the Gaussian linear regression with unit error variance, $n=500$, $\kappa=p/n=0.3$, a case in the context of bootstrap for high-dimensional problems in \cite{el2018can}. Plot (a) shows the trajectories of 10 SA runs for each of 10 equally-spaced targeted coverage probabilities $0.05, 0.15, ..., 0.95$. The trajectories of 10 $m$-values in terms of $m/n$ and $m/n - 1$ are displayed in dotted lines (compressed to half its original length by displaying the values every two iterations for clarity). The solid lines denote the cumulative mean values of $F_\theta(T_{y,\theta})$.
        The circle ``$\circ$'' points in Plot (b) is the Q-Q plot of 1,000 bootstrap values of $(\beta_*-\hat{\beta})'(X'X)(\beta_*-\hat{\beta})/p$ against the theoretical quantiles. The triangle ``$\triangle$'' points show the corresponding quantiles for those obtained by bootstrap with the 10 estimated $m$ values. The gray vertical lines denote the 10 quantiles correspond to the 10 targeted coverage probabilities.
        }
\label{fig:lrParUnitVar01}
\end{figure}
}{}

\ifthenelse{\boolean{inline}}{
\begin{table}
\caption{Upper bounds of the constructed confidence Interval for $(\beta-\hat{\beta})'(X'X)(\beta-\hat{\beta})/p$ at various significance levels $\alpha$, comparing the standard bootstrap and CB. Each method uses 1,000 bootstrap samples, with the number of observations $m$ resampled from the original dataset (size $n = 500$), as determined in the RA step, specified in brackets. All intervals have lower bounds fixed at $0$.\label{tab:1} }
\vspace{0.1 in}
\footnotesize
\centering
\begin{tabular}{ccccccccccc}
\hline
\hline
              & 0.05                                                      & 0.15                                                      & 0.25                                                      & 0.35                                                      & 0.45                                                      & 0.55                                                      & 0.65                                                      & 0.75                                                      & 0.85                                                      & 0.95                                                      \\ \hline
Truth         & 1.197                                                     & 1.120                                                     & 1.075                                                     & 1.041                                                     & 1.010                                                     & 0.981                                                     & 0.952                                                     & 0.920                                                     & 0.881                                                     & 0.818                                                     \\
Bootstrap & \begin{tabular}[c]{@{}c@{}}1.820\\ {[}500{]}\end{tabular} & \begin{tabular}[c]{@{}c@{}}1.650\\ {[}500{]}\end{tabular} & \begin{tabular}[c]{@{}c@{}}1.547\\ {[}500{]}\end{tabular} & \begin{tabular}[c]{@{}c@{}}1.454\\ {[}500{]}\end{tabular} & \begin{tabular}[c]{@{}c@{}}1.385\\ {[}500{]}\end{tabular} & \begin{tabular}[c]{@{}c@{}}1.324\\ {[}500{]}\end{tabular} & \begin{tabular}[c]{@{}c@{}}1.262\\ {[}500{]}\end{tabular} & \begin{tabular}[c]{@{}c@{}}1.190\\ {[}500{]}\end{tabular} & \begin{tabular}[c]{@{}c@{}}1.107\\ {[}500{]}\end{tabular} & \begin{tabular}[c]{@{}c@{}}0.992\\ {[}500{]}\end{tabular} \\
CB            & \begin{tabular}[c]{@{}c@{}}1.176\\ {[}651{]}\end{tabular} & \begin{tabular}[c]{@{}c@{}}1.111\\ {[}630{]}\end{tabular} & \begin{tabular}[c]{@{}c@{}}1.074\\ {[}617{]}\end{tabular} & \begin{tabular}[c]{@{}c@{}}1.032\\ {[}611{]}\end{tabular} & \begin{tabular}[c]{@{}c@{}}1.006\\ {[}603{]}\end{tabular} & \begin{tabular}[c]{@{}c@{}}0.984\\ {[}597{]}\end{tabular} & \begin{tabular}[c]{@{}c@{}}0.938\\ {[}598{]}\end{tabular} & \begin{tabular}[c]{@{}c@{}}0.919\\ {[}584{]}\end{tabular} & \begin{tabular}[c]{@{}c@{}}0.896\\ {[}569{]}\end{tabular} & \begin{tabular}[c]{@{}c@{}}0.801\\ {[}557{]}\end{tabular} \\ \hline
\end{tabular}
\end{table}

}{}

For a simulation study, we took $n=500$ and $\kappa=p/n=0.3$, a case in the context of bootstrap for high-dimensional problems as discussed in \cite{el2018can}. We set $\sigma=1$ to be known. Our goal is to conduct a valid and efficient joint inference on $\beta$ through resampling. Our numerical experiments indicate that both standard bootstrap and the residual bootstrap \citep{Freedman1981} yield unsatisfactory results in cases where $\kappa=p/n=0.3$, with variations in the choice of $n$ (see Supplementary~S.7 for details). This observation aligns with the results reported in \cite{bickel1983bootstrapping} and \cite{el2018can}.
Here, we show that the proposed CB method can give a valid and efficient joint inference on $\beta$.  The SA algorithm to find the optimal $m$ was applied to each of 10 equally spaced target coverage probabilities $0.05, 0.15, ..., 0.95$. 
The trajectories of the 10 SA runs and the corresponding 10 estimated $m$-values are shown in Figure \ref{fig:lrParUnitVar01} (a). We see that all estimated $m$ values are larger than $n$ in this case. We compare the constructed frequentist $1-\alpha$ confidence intervals for $(\beta-\hat{\beta})'(X'X)(\beta-\hat{\beta})/p$ with standard bootstrap and the estimated $m$-out-of-$n$ bootstrap by RA with the 10 confidence coefficients $\alpha$. Both bootstrap confidence intervals are constructed by taking $0$ as the lower bound and $1-\alpha$ quantile of the resampled values $(\tilde{\hat{\beta}}- \hat{\beta})'(X'X)(\tilde{\hat{\beta}}-\hat{\beta})/p$  as the upper bound. As a matter of fact, similar to the reasoning in Section~\ref{ss:ra-simple}, it can be shown that in this example, this interval construction method for the proposed RA is equivalent to (\ref{eq:CI-region}). The results are summarized in Table~\ref{tab:1}. It can be seen that our proposed method gives an estimate numerically very close to the true confidence region derived theoretically.

Figure \ref{fig:lrParUnitVar01} (b) compares the distribution of the standard
 bootstrap estimates of the confidence distribution $(\beta_*-\hat{\beta})'(X'X)(\beta_*-\hat{\beta})/p$ 
 against the theoretical distribution obtained from (\ref{eq:fiducial-theoretical-2}).  The corresponding distribution obtained by the $m$-out-of-$n$ bootstrap with the mixture of 10 estimated values of $m$ (each with equal probabilities) is also shown. It can be seen that the resampling scheme found by RA dramatically outperforms the standard bootstrap in recovering the theoretical distribution of $(\beta_*-\hat{\beta})'(X'X)(\beta_*-\hat{\beta})/p$. However, it appears challenging to perfectly recover the desired theoretical distribution via a mixture of the $m$-out-of-$n$ bootstrap, regardless of the number of mixture components. This suggests the potential for exploring alternative resampling schemes that can provide greater flexibility than the $m$-out-of-$n$ bootstrap. Furthermore, this leads to our 
 proposed refinement method that is discussed next
 in Section \ref{ss:SA-fiducial}.
        
\end{example}

\subsection{Distributional Resampling of $T_{y, \theta}$}
\label{ss:SA-fiducial}

\ifthenelse{1=1}{}{
\begin{remark}[An motivating example... FIXME. {\color{red} This seems to provide deep intuitions on why bootstrap can fail!}]
\label{remark:conjecture}
    Consider $m$-out-of-$n$ bootstrap for the linear regression model in Example \ref{example:lr}.
    There is no random variable $m$ such that the bootstrap distribution is (close to) the fidicual distribution.
    
\begin{proof}
    The variability of bootstrapped $(X'X)^{-1}$, {\it, that is,}, $(X_*'X_*)^{-1}$, has more variability than $(X'X)^{-1}$ and its scaled variants. More precisely, the shape of the ellipse $(X_*'X_*)^{-1}$ varies and is different from that of $(X'X)^{-1}$. This difference cannot be eliminated by allowing $m$ to take different values.
\end{proof}
\end{remark}

{\color{red} Do we use the word {\it fiducial} in the sense of matching observed with predicted?}
}

As elaborated in Section~\ref{ss:SA}, to approximate the confidence distribution of $\theta$ with bootstrapped $\hat{\theta}_*$'s, it requires that (\ref{eq:fiducial-dist}) holds. This motivates a simple refinement method to resample the bootstrapped estimates $\hat{\theta}_*$'s in (\ref{eq:resampling-1}) obtained in the RA step to create a new distribution such that the values of $F_{\hat{\theta}_*} (T_{y,\hat{\theta}_*})$ with the resampled $\hat{\theta}_*$'s closely approximate a sample from the standard uniform distribution. Such resampled distribution can then be seen as an approximation to the targeted confidence distribution and can be used for constructing confidence sets at all $\alpha \in (0,1)$ levels by Theorem~\ref{thm:2}. The proposed method is summarized into the following three-step algorithm, which is referred to as {\it distributional resampling} (DR). 

\addtocounter{alg}{+1}
\begin{alg}[Distributional Resampling]    
Suppose that a set of $m$ values for the $m$-out-of-$n$ bootstrap method is obtained by applying the SA algorithm with varying target coverage level $\alpha$ (e.g. $\alpha = \{0.05, 0.50, 0.95\}$). The DR algorithm creates a sample of selected bootstrap estimates in three steps:
\begin{enumerate} 
    \item Create $B$ bootstrapped estimates $\hat{\theta}_*^{(1)}, \dots, \hat{\theta}_*^{(B)}$ with sufficiently large $B$ using the $m$ values obtained by the SA algorithm. 
    \item Compute $U_b = F_{\hat{\theta}_*^{(b)}}(T_{y, \hat{\theta}_*^{(b)}})$ with Monte Carlo approximation in (\ref{eq:fiducial-constraint-04}) for $b = 1,\dots,B$.
    \item Set $\tilde{\Theta} = \emptyset$ as the set of resampled estimates and repeat the following steps for $B$ times:
    \begin{enumerate}
    \item Draw $u\sim \mbox{Uniform}(0,1)$;
    \item Find $b_* = \arg\min_{b}\left|U_b- u\right|$;
    \item Add $\hat{\theta}_*^{(b_*)}$ to set $\Tilde{\Theta}$. 
    \end{enumerate}
\end{enumerate}
\end{alg}

\begin{remark}
    The DR algorithm can be employed with any sets of $\hat{\theta}_*^{(1)}, \dots, \hat{\theta}_*^{(B)}$. Nonetheless, the RA step, which offers a close, albeit not flawless, approximation to the true confidence distribution of $\theta$, can produce a more optimal effective sample size, making it a more effective choice. The example in Supplementary~S.4 shows that the construction of the confidence distribution via DR can be practically impossible without the RA step. The effect of the varying $\alpha$ levels used in the RA step on the final result is also empirically studied in Supplementary~S.4, showing that the result is robust to the grid density of $\alpha$. Another example in \cite{martin2023b} where the true confidence distribution is known is used to demonstrate the efficiency of the approximation in step (c) of the DR algorithm (see Supplementary~S.8).
\end{remark}

\begin{remark}
 For computational efficiency, the (a) and (b) steps of the DR algorithm can be combined into the SA algorithm in the RA step. Specifically, the bootstrapped estimates $\hat{\theta}_*^{(b)}$ as well as the value $ F_{\hat{\theta}_*^{(b)}}(T_{y, \hat{\theta}_*^{(b)}})$ are collected along the SA procedure. The resulting algorithm is given in Supplementary~S.9.
\end{remark}


\ifthenelse{\boolean{inline}}{
\begin{figure}
\centering
\includegraphics[width=5.5in]{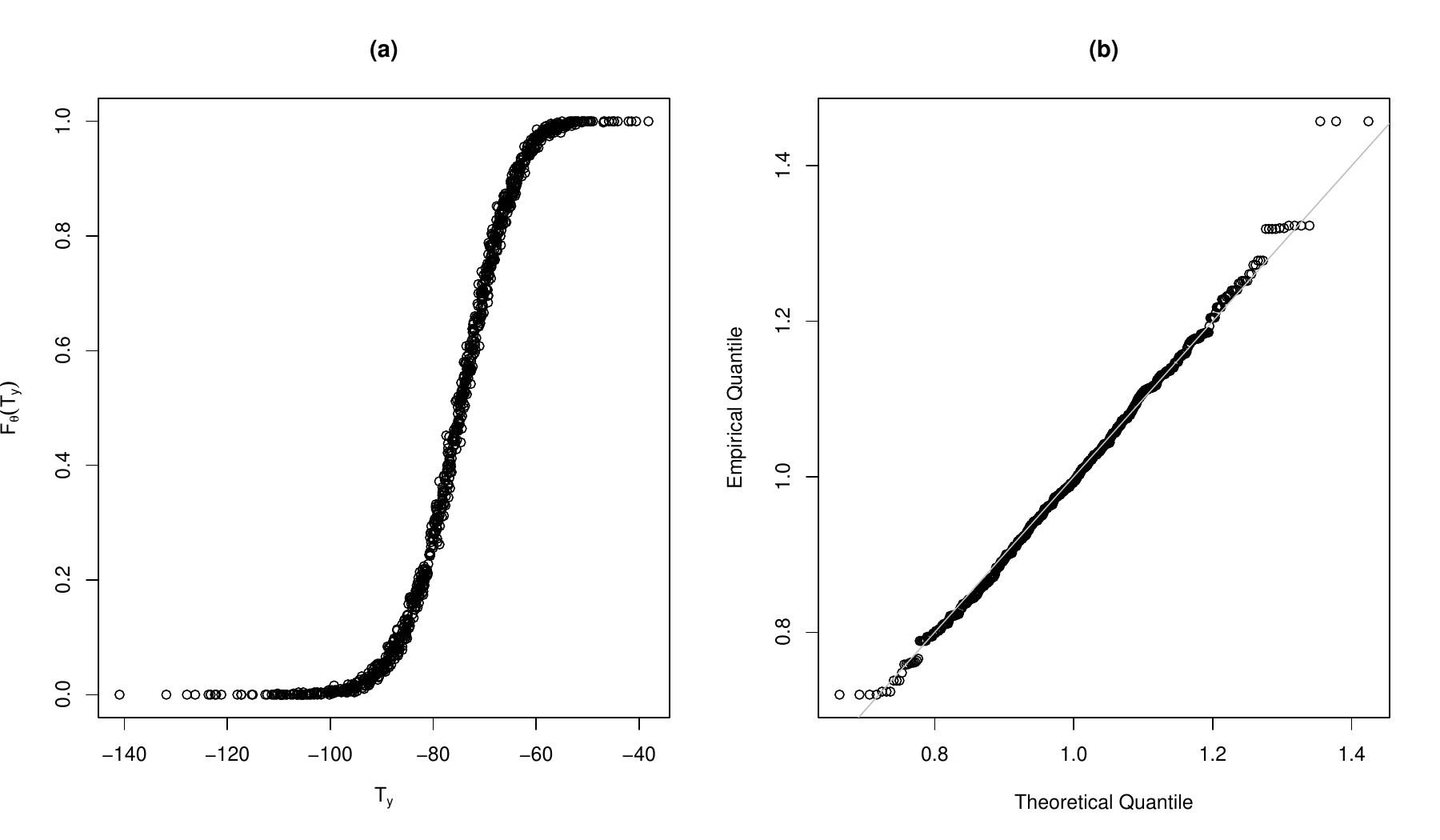}
        \caption{Illustration of the DR step after RA.
        Plot (a) shows the estimate $\hat{F}_{\hat{\theta}_*}(.)$ via Monte Carlo with 1,000 $\hat{\theta}_*$ values obtained by bootstrapping with the 10 estimated $m$ values explained in Figure \ref{fig:lrParUnitVar01}. Plot (b) is the Q-Q plot of the bootstrap values of $(\beta_*-\hat{\beta})'(X'X)(\beta_*-\hat{\beta})/p$ selected by DR against the theoretical quantiles.
        }
\label{fig:lrParUnitVar02}
\end{figure}
}{}

\addtocounter{example}{-1}
\begin{example}[Cont'd]
\rm
We applied the DR method to the bootstrapped estimates of $\theta$ obtained with the $m$ values from the RA process. Figure \ref{fig:lrParUnitVar02} (a) shows the scatter plot of the initial
$m$-out-of-$n$ bootstrapped values of 
$F_{\hat{\theta}_*}(T_{y, \hat{\theta}_*})$ versus the corresponding values of $T_{y, \hat{\theta}_*}$.
Similar to Figure \ref{fig:lrParUnitVar01} (b), Figure \ref{fig:lrParUnitVar02} (b)
shows the empirical quantiles of the resampled estimates $(\beta_*-\hat{\beta})'(X'X)(\beta_*-\hat{\beta})/p$ obtained with DR versus the theoretical quantiles. Comparing Figure \ref{fig:lrParUnitVar02} (b) to Figure \ref{fig:lrParUnitVar01} (b), we see that the application of RA followed by DR creates a nearly perfect approximation to the true confidence distribution. 
\end{example}



\subsection{Marginal Parametric Inference}\label{ss:marginal}

In the previous section, it was shown that the proposed RA method followed by the DR refinement process can provide a valid inference on the model parameters $\theta \in \Theta$. In practice, we might only be interested in the marginal inference of some functions of $\theta$. For example, when partitioning the model parameter as $\theta = (\theta_1, \theta_2) \in \Theta=\Theta_1\times \Theta_2$, we are interested in constructing a confidence region for $\theta_1$, while treating $\theta_2$ as the nuisance parameter. Here, we propose a solution inspired by \cite{martin2015plausibility, martin2023valid}, which replaces the generalized association function (\ref{eq:T-association-01}) with a form of the relative profile likelihood
\begin{equation}\label{eq:T-association-02}
    T_{y,\theta_1} =  \ell(y, \hat\theta_y) - \max_{\theta_2}\ell(y,\theta_1,\theta_2), \qquad\theta\in\Theta, y\in {\mathbb Y}.
\end{equation}
The idea is to create an association function that depends on $\theta_1$ alone, by marginalizing out $\theta_2$ via the profile likelihood. 

\ifthenelse{\boolean{inline}}{
\begin{figure}
\centering
\includegraphics[width=5.5in]{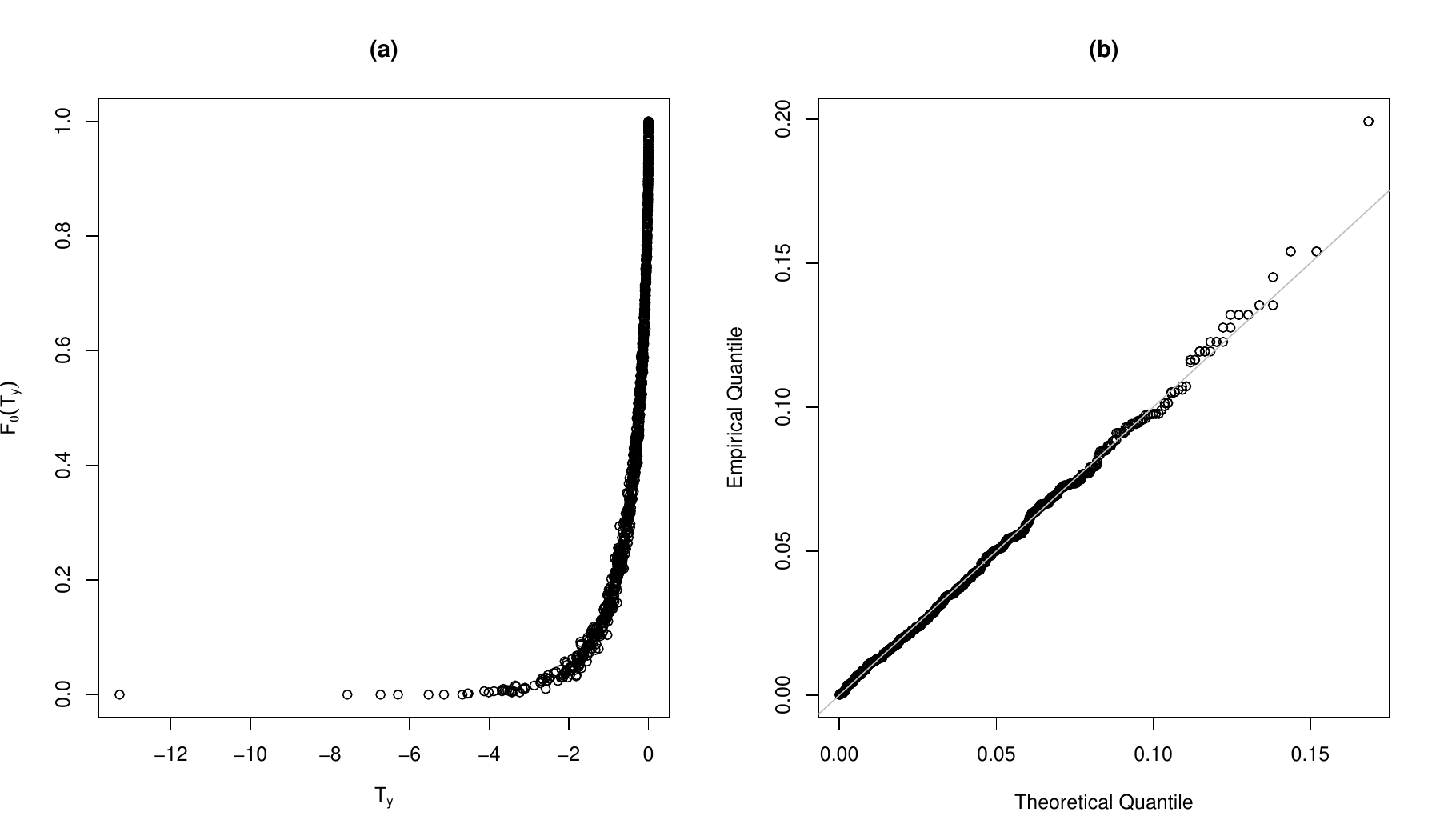}
        \caption{Illustration of the results after RA and DR for marginal inference of $\beta_1$. Plot (a) shows the estimate $\hat{F}_{\hat{\theta}_*}(.)$ via Monte Carlo with 1,000 $\hat{\theta}_*$ values obtained by bootstrapping with the 10 estimated $m$ values through the RA step. Plot (b) is the Q-Q plot of the bootstrap values of $|\beta_{1}^*-\hat{\beta}_1|$ selected by DR against the theoretical quantiles.
        }
\label{fig:lr-marginal}
\end{figure}
}

In our proposed CB method, we conduct the RA process using the generalized association function defined in (\ref{eq:T-association-02}). If the distribution of $T_{Y,\theta_1}$ as a function of $Y\sim \mathbb{P}_{\theta}$ is independent of $\theta_2$, the exactness of the constructed confidence region for $\theta_1$ via $\{\theta_1: F_{\theta_1}(T_{y,\theta_1}) \geq \alpha\}$ for $\alpha \in (0,1)$, as previously elaborated in Theorem~\ref{thm:exact} is straightforward to establish \citep[see, {\it e.g.},][]{martin2015plausibility}. Even if when $T_{y,\theta_1}$ being independent of $\theta_2$ is hard to verify, \cite{martin2023valid} demonstrates that inference conducted using the profile likelihood-based association function (\ref{eq:T-association-02}) remains both theoretically valid and empirically efficient. This robust theoretical foundation further supports the validity of the proposed CB approach. 
\addtocounter{example}{-1}
\begin{example}[Cont'd]
\rm
We considered the marginal inference of the individual coefficient $\beta_1$. The RA and DR steps are carried out with the relative profile likelihood $\eqref{eq:T-association-02}$, where $\theta_1 := \beta_1$ and $\theta_2 := \{\beta_j\, | \, j \neq 1\} $. The results are shown in Figure~\ref{fig:lr-marginal}. It can be seen that the CB method provides a very good approximation to the true fiducial distribution of $|\beta^*_{1} - \hat{\beta}_1|$. 
\end{example}

\section{Applications in $L_1$-penalized Linear Regression}\label{s:application}
Revisit the linear regression model (\ref{eq:LR01}) articulated in Example~\ref{example:lr}. The $L_1$-penalized estimator of $\beta$, also known and popularized as the Lasso estimator \citep{tibshirani1996}, is considered here and can be written as 
\[
\hat{\beta} = \argmin_\beta \frac{1}{2}(Y - X\beta)^T(Y - X\beta) + \lambda \sum_{i = 1}^p \vert \beta \vert_i,
\]
where $\lambda > 0$ is the tuning parameter. The Lasso estimator is useful especially when $p > n$, as it can provide variable selection by shrinking some coefficients toward zero. The involvement of the regularization term, allowing for a more robust estimator compared to the ordinary least squares (OLS) estimator, can result in better prediction power.

Although Lasso is widely used in real data analysis scenarios such as genome-wide association studies (GWAS, \citeauthor{Uffelmann2021}, \citeyear{Uffelmann2021}; \citeauthor{Zeng2015}, \citeyear{Zeng2015}), designing a valid and efficient inference procedure is generally considered difficult, due to the added penalty term. Resampling-based approaches, such as bootstrap and permutation tests \citep{Arbet2017}, are the most commonly used methods for model inference. However, previous theoretical studies have shown that even in the asymptotic case, the validity of the bootstrap-based inference procedure for Lasso is difficult to establish \citep{Fu2000,Chatterjee2010,Chatterjee2011}. 

A remarkable benefit of the proposed CB method lies in its ability to offer a valid frequentist inference procedure for methods like Lasso, where valid inference is challenging to derive mathematically. We will first use a simulation study to demonstrate the efficiency of the proposed method in Section \ref{ss:L1-simu} and then compare our proposed method with other methods in a real data example in Section \ref{ss:L1-data}.

\subsection{Simulation Study}\label{ss:L1-simu}

Similar to the high-dimensional setting in Example~\ref{example:lr}, our simulation studies consider the scenario where the ratio $\kappa =  p/n \to 1$. Specifically, we increase $\kappa$ by varying its value with $\{0.3, 0.5, 0.9 \}$, alongside an increase in sample size $n$ with $\{100, 200, 500\}$. We let $X_{ij} \sim N(0,1)$ and standardize $X$ to have mean 0 and standard deviation 1 in each column. For the vector of true parameters $\beta$, we assume the signal to be sparse by setting $\beta = (3, 0, \dots,0)$. In our data simulation, we set $\sigma = 1.$ When performing inference with CB, both $\beta$ and $\sigma$ are treated as unknown. In the case of Lasso, a key concern is the valid inference of $\beta$. We consider both the joint inference on $\beta$ and the marginal inference on $\beta_1$ (the true signal). To compare with other bootstrap-based methods, we include the standard bootstrap \citep{efron1979bootstrap} as well as the (debiased) residual bootstrap \citep{Chatterjee2010}. Specifically, carefully designed residual bootstrap has been shown to be particularly effective in scenarios where $p/n \to 1$ \citep{Lopes2014}, offering a powerful alternative to the standard bootstrap.

To perform the proposed CB, we take $\pi(\theta) = \lambda \sum_{i = 1}^p \vert \beta \vert_i$ in (\ref{eq:loss}) to formulate the association function $T_{y, \theta}$ with penalized likelihood. The penalty term $\lambda$ used for regularization takes values of $\lambda \in \{20.1, 40.2, 63.1\}$ correspondingly for the three cases, which is initially determined by 10-fold cross-validation on the entire observed dataset, and remains constant across all 500 repetitions. We fix $\sigma^2$ at $\hat{\sigma}^2 = \frac{(\beta-\hat{\beta}_y)'(X'X)(\beta-\hat{\beta}_y)}{n-\sum_{i=1}^p\mathbb{I}\{\hat{\beta }_y\neq 0\} }$, a refined estimator given in \cite{Reid2016}, where $\hat{\beta}_y$ denotes the Lasso estimates of $\beta$ on the entire observed dataset. The implementation of CB uses the R package $\texttt{glmnet}$ \citep{Friedman2010} and $\texttt{natural}$ \citep{Yu2019}.


A valid joint inference aims to provide a close approximation of the true distribution of the quantity $(\beta-\hat{\beta}_y)'(X'X)(\beta-\hat{\beta}_y)/p$. We apply the RA step followed by the DR step the create an approximate confidence distribution for the parameter $\beta$, and construct the $1-\alpha$ confidence region $\{\beta:  (\beta-\hat{\beta}_y)'(X'X)(\beta-\hat{\beta}_y)/p \leq q_{1-\alpha}\}$, where $q_{1-\alpha}$ is determined in such a way that $100(1-\alpha)\%$ of the samples $\hat{\beta}_*$ obtained from the DR step fall within the region. We refer to $q_{1-\alpha}$ as the magnitude of the confidence region, analogous to the length of the confidence interval in a multidimensional context.  The summarized results, presented in Table~\ref{tab:2}, demonstrate that CB can achieve the desired joint coverage rate while the standard bootstrap exhibits significant undercoverage. The residual bootstrap also performs effectively in such $p/n \to 1$ simulation setting, as validated in \cite{Lopes2014}; however, it tends to yield more conservative intervals than CB.

\ifthenelse{\boolean{inline}}{
\begin{table}
\caption{\label{tab:2}Estimated coverage probabilities of the true $\beta$ and expected magnitude for the constructed $100(1-\alpha)\%$ joint confidence region with different values of $\alpha$ using the model settings ($p/n \to 1$). The standard deviation of each value (estimated with bootstrap) is given in parentheses.}
\vspace{0.1 in}
\centering
\scriptsize
\begin{tabular}{cccccccccc}
\hline
                                                                                 &          & \multicolumn{2}{c}{CB}        &  & \multicolumn{2}{c}{Standard Bootstrap} &  & \multicolumn{2}{c}{Residual Bootstrap} \\ \cline{3-4} \cline{6-7} \cline{9-10} 
Model                                                                            & $\alpha$ & Coverage      & Magnitude     &  & Coverage           & Magnitude         &  & Coverage           & Magnitude         \\ \hline
\multirow{3}{*}{\begin{tabular}[c]{@{}c@{}}$n = 100$, \\ $p = 30$\end{tabular}}  & 0.05     & 0.952 (0.010) & 0.463 (0.002) &  & 0.780 (0.019)      & 0.270 (0.003)     &  & 0.958 (0.009)      & 0.471 (0.002)     \\
                                                                                 & 0.15     & 0.862 (0.015) & 0.328 (0.001) &  & 0.604 (0.022)      & 0.183 (0.002)     &  & 0.872 (0.015)      & 0.334 (0.001)     \\
                                                                                 & 0.25     & 0.780 (0.019) & 0.259 (0.001) &  & 0.486 (0.022)      & 0.142 (0.002)     &  & 0.786 (0.018)      & 0.265 (0.001)     \\ \hline
\multirow{3}{*}{\begin{tabular}[c]{@{}c@{}}$n = 200$, \\ $p = 100$\end{tabular}} & 0.05     & 0.946 (0.010) & 0.201 (0.001) &  & 0.446 (0.022)      & 0.077 (0.001)     &  & 0.962 (0.009)      & 0.205 (0.000)     \\
                                                                                 & 0.15     & 0.864 (0.015) & 0.152 (0.000) &  & 0.298 (0.020)      & 0.052 (0.001)     &  & 0.864 (0.015)      & 0.154 (0.000)     \\
                                                                                 & 0.25     & 0.758 (0.019) & 0.125 (0.000) &  & 0.192 (0.018)      & 0.040 (0.001)     &  & 0.762 (0.019)      & 0.127 (0.000)     \\ \hline
\multirow{3}{*}{\begin{tabular}[c]{@{}c@{}}$n = 500$,\\  $p = 450$\end{tabular}} & 0.05     & 0.950 (0.010) & 0.045 (0.000) &  & 0.844 (0.016)      & 0.034 (0.000)     &  & 0.946 (0.010)      & 0.046 (0.000)     \\
                                                                                 & 0.15     & 0.866 (0.015) & 0.034 (0.000) &  & 0.726 (0.020)      & 0.026 (0.000)     &  & 0.876 (0.015)      & 0.035 (0.000)     \\
                                                                                 & 0.25     & 0.762 (0.019) & 0.028 (0.000) &  & 0.646 (0.021)      & 0.023 (0.000)     &  & 0.778 (0.019)       & 0.029 (0.000)     \\ \hline
\end{tabular}
\end{table}
}{}

\ifthenelse{\boolean{inline}}{
\begin{table}
\caption{\label{tab:3}Estimated coverage probabilities of the true $\beta_1$ and expected length for the constructed $100(1-\alpha)\%$ marginal confidence interval with different values of $\alpha$ using the model settings ($p/n \to 1$). The standard deviation of each value (estimated with bootstrap) is given in parentheses.}
\vspace{0.1 in}
\centering
\scriptsize
\begin{tabular}{cccccccccc}
\hline
                                                                                 &          & \multicolumn{2}{c}{CB}        &  & \multicolumn{2}{c}{Standard Bootstrap} &  & \multicolumn{2}{c}{Residual Bootstrap} \\ \cline{3-4} \cline{6-7} \cline{9-10} 
Model                                                                            & $\alpha$ & Coverage      & Length        &  & Coverage           & Length            &  & Coverage           & Length            \\ \hline
\multirow{3}{*}{\begin{tabular}[c]{@{}c@{}}$n = 100$, \\ $p = 30$\end{tabular}}  & 0.05     & 0.950 (0.010) & 0.735 (0.002) &  & 0.510 (0.022)      & 0.409 (0.002)     &  & 0.950 (0.010)      & 0.739 (0.001)     \\
                                                                                 & 0.15     & 0.844 (0.016) & 0.615 (0.002) &  & 0.320 (0.021)      & 0.298 (0.001)     &  & 0.860 (0.016)      & 0.616 (0.001)     \\
                                                                                 & 0.25     & 0.764 (0.019) & 0.544 (0.001) &  & 0.212 (0.018)      & 0.238 (0.001)     &  & 0.756 (0.019)      & 0.544 (0.001)     \\ \hline
\multirow{3}{*}{\begin{tabular}[c]{@{}c@{}}$n = 200$, \\ $p = 100$\end{tabular}} & 0.05     & 0.940 (0.011) & 0.634 (0.001) &  & 0.218 (0.018)      & 0.287 (0.001)     &  & 0.958 (0.009)      & 0.639 (0.001)     \\
                                                                                 & 0.15     & 0.858 (0.016) & 0.549 (0.001) &  & 0.072 (0.012)      & 0.210 (0.001)     &  & 0.862 (0.015)      & 0.553 (0.000)     \\
                                                                                 & 0.25     & 0.746 (0.019) & 0.499 (0.001) &  & 0.038 (0.009)      & 0.167 (0.001)     &  & 0.756 (0.019)      & 0.501 (0.000)     \\ \hline
\multirow{3}{*}{\begin{tabular}[c]{@{}c@{}}$n = 500$,\\  $p = 450$\end{tabular}} & 0.05     & 0.942 (0.010) & 0.397 (0.001) &  & 0.188 (0.017)      & 0.178 (0.000)     &  & 0.938 (0.011)      & 0.400 (0.000)     \\
                                                                                 & 0.15     & 0.852 (0.016) & 0.344 (0.000) &  & 0.086 (0.013)      & 0.131 (0.000)     &  & 0.862 (0.015)      & 0.345 (0.000)     \\
                                                                                 & 0.25     & 0.750 (0.019) & 0.313 (0.000) &  & 0.048 (0.010)      & 0.104 (0.000)     &  & 0.760 (0.019)      & 0.313 (0.000)     \\ \hline
\end{tabular}
\end{table}
}{}

We also performed the marginal inference on $\beta_1$. The confidence interval is constructed as $\{\beta_1 : |\beta_1 - \hat{\beta}_{y,1}| \leq q_{1-\alpha}\}$ where $q_{1-\alpha}$ is chosen such that $100(1-\alpha)\%$ of the samples $\hat{\beta}^*_1$ obtained from the DR step are covered by the interval. The summarized results are shown in Table~\ref{tab:3}, demonstrating that CB can achieve the desired marginal coverage rate. In contrast, the standard bootstrap method continues to exhibit significant undercoverage.


\subsection{Real Data Example}\label{ss:L1-data}

Consider the diabetes study introduced by \cite{efron2004least}. This study encompasses ten baseline variables ($p = 10$), including age, sex, body mass index (BMI), mean arterial pressure (MAP), and six blood serum measurements (S1-S6).  The primary response variable of interest corresponds to a quantitative measure of disease progression, recorded one year after baseline assessment, for each of the $n = 442$ diabetes patients in the dataset. 

We start by standardizing the variables so that $\sum_{i=1}^n x_{ij} = 0, \frac{1}{n}\sum_{i=1}^n x^2_{ij} = 1, \mbox{for } j = 1,\dots,p$ and $\frac{1}{n}\sum_{i=1}^n y_i = 0$. We also performed a routine linear regression diagnostic analysis to verify that the basic assumptions of the linear regression model are appropriate. Then, we fit a Lasso version of the model to the dataset for predicting the response variable of interest, perform variable selection, and make inference on regresion coefficients of the predictors. The penalty value $\lambda \approx 520$ is selected by 10-fold cross-validation, and $\sigma^2$ is estimated using the residual sum of squares from the standard linear regression model including all predictors. 

\ifthenelse{\boolean{inline}}{
\begin{figure}
\centering
\includegraphics[width=5in]{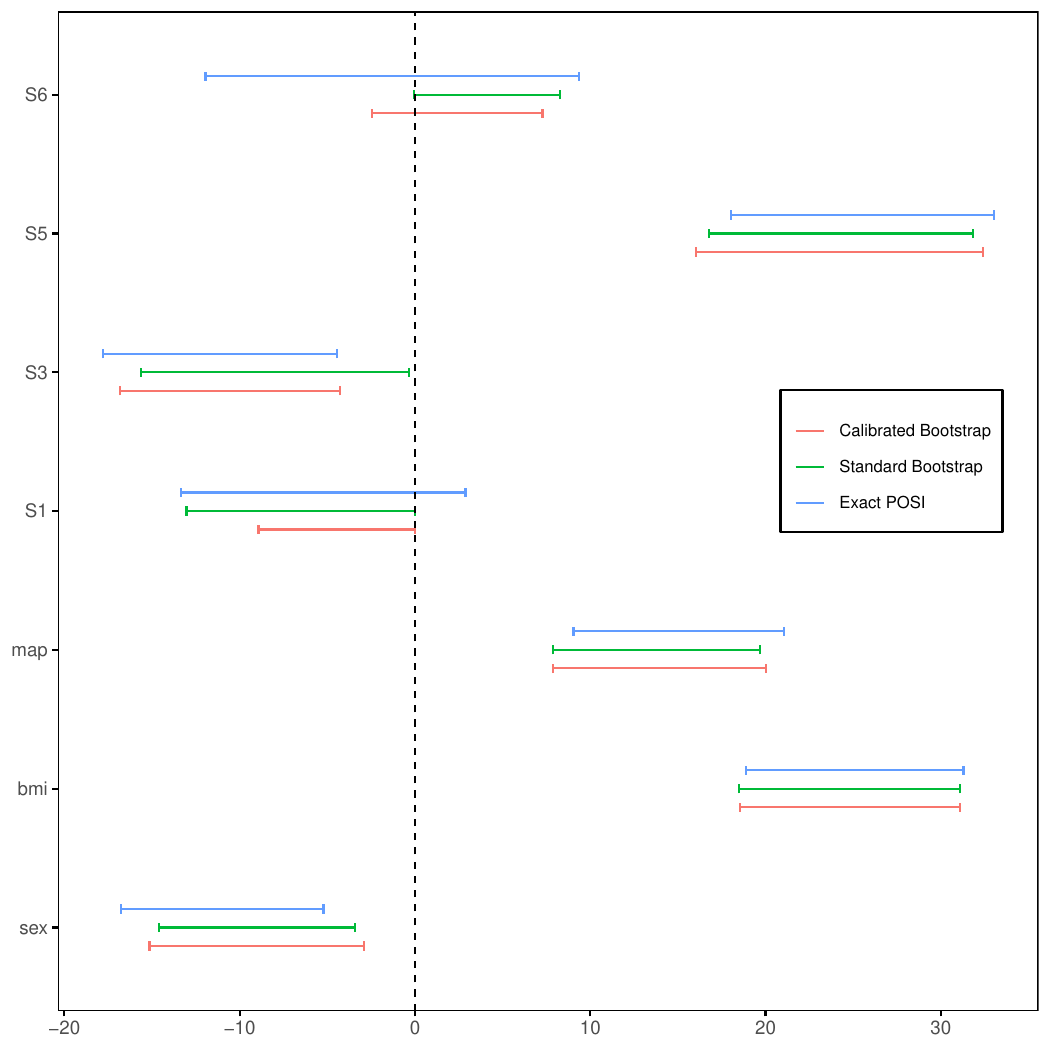}
        \caption{Constructed 95\% confidence intervals for the seven selected variables by Lasso ($\lambda = 520$). The intervals constructed with our proposed CB method are shown in red lines. The intervals constructed with the standard bootstrap are shown in green lines. Exact POSI denotes the post-selection inference approach of \cite{Lee2016}, and the constructed intervals are shown in blue lines.
        }
\label{fig:diabetes}
\end{figure}
}{}

Lasso yields a model with seven variables: sex, BMI, MAP, S1, S3, S5 and S6. For these selected variables, we compare the confidence intervals obtained by our proposed CB method with those obtained by the standard bootstrap. We also include a recent method of \cite{Lee2016} designed for post-selective inference of Lasso. The method, which we refer to as exact POSI, is capable of constructing a valid confidence interval for the predictors selected by Lasso, while conditioning on the selected model. The constructed 95\% intervals with the three methods are shown in Figure~\ref{fig:diabetes}. For the three methods, our method ensures a minimum of $95\%$ coverage under the full model with 10 predictors, in contrast to exact POSI, which guarantees this level of coverage conditional on the model with seven selected predictors. The standard bootstrap does not warrant $95\%$ coverage. The intervals from our CB method are shown to be comparable in length to those of other methods but are notably shorter than the exact POSI for S6. The results also indicate a consensus among all methods regarding the non-significance of S6. 

\section{Concluding Remarks}
\label{s:discussion}
In this paper, we proposed a resampling approximation approach that enables valid finite sample joint inference based on likelihood functions and marginal inference based on profile likelihood. It can be easily extended to cases where general loss functions are used for point estimation. Avoiding the limitations associated with conventional bootstrap techniques, the proposed method is shown to achieve valid inference outcomes through calibrated resampling and refinements. 
To our knowledge, this is the first study to adapt the resampling method for valid inference in finite-sample scenarios. 
Although the proposed method does involve higher computational costs compared to conventional bootstrap techniques, it remains computationally feasible. Moreover, it can be readily parallelized on modern computer clusters, further enhancing its computational feasibility and efficiency.

The key idea in our development of resampling-based methods for finite-sample valid inference is to find a resampling scheme that guarantees the validity of resulting confidence region for a pre-specified confidence level. In the proposed CB method, we created a stochastic approximation algorithm to find an adaptive $m$-out-of-$n$ resampling scheme. Alternative resampling schemes and alternative ways of creating samples of bootstrapped estimates can be considered in future research and applications. For example, weighted maximum likelihood estimates with weights drawn from an adaptive $\mbox{Dirichlet}(\delta {\mathbf 1})$ distribution can be considered, with $\delta$ representing the calibration parameter. This method is expected to be especially useful for the case with small observed data samples. 

Nevertheless, it is worth noting that obtaining exact marginal inferences for individual parameters can be a complex undertaking, necessitating further research. This highlights the broader challenges of conducting finite-sample valid inference in complex models. In this article, we proposed an approach based on the idea in \cite{martin2015plausibility,martin2023valid} for ``marginalizing out'' nuisance parameters. While this method is shown to be valid and empirically efficient, it points to a potential direction that invites creative thoughts on marginal inference, a challenging problem for all existing schools of thought.

The primary focus of this paper is the development of computational methods that facilitate efficient parametric inference. \cite{cella2022direct} developed an IM framework for approximate inference on risk minimizers in a nonparametric context. Given that many modern machine learning applications necessitate inference based on unknown models or loss functions, an intriguing future direction would be to explore whether the insights from this paper can be effectively adapted for efficient inference under the framework proposed by \cite{cella2022direct}.

\section*{Acknowledgements}
The authors are grateful to the editor, the associate editor, and anonymous referees for their insightful, critical, and constructive comments on an earlier version of the paper.
Zhang and Jiang are supported in part by U.S. National Institutes of Health grants R01HG010171 and R01MH116527 and National Science Foundation grant DMS-2112711. Liu is supported in part by U.S. National Science Foundation grant DMS-2412629.

\bibliographystyle{rss}
\bibliography{ims}

\end{document}


\def\spacingset#1{\renewcommand{\baselinestretch}%
{#1}\small\normalsize} \spacingset{1}

\newcommand{\bX}{{\mathbf X}}
\newcommand{\bY}{{\mathbf Y}}
\newcommand{\bm}[1]{\boldsymbol{#1}}


\renewcommand{\bm}[1]{#1}
\newcommand{\bx}{\boldsymbol{x}}
\newcommand{\by}{\boldsymbol{y}}


\addtolength{\textheight}{.5in}%

\appendix
\renewcommand{\thesection}{S}
\renewcommand{\theequation}{S.\arabic{equation}} 
\setcounter{equation}{0}

\renewcommand{\thefigure}{S.\arabic{figure}}

\setcounter{figure}{0}

\renewcommand{\thetable}{S.\arabic{table}}
\setcounter{table}{0}

\renewcommand{\thealgocf}{S.\arabic{algocf}} 
\setcounter{algocf}{0} 

  \bigskip
  \bigskip
  \bigskip
  \begin{center}
    {\Large {\bf Finite Sample Valid Inference via Calibrated Bootstrap}:\\

\vspace{0.2in}

    \textit{\Large Supplementary Material: Proofs, Discussions \\ and Additional Experiments}
}

\vspace{0.1in}

\text{Yiran Jiang} \\
\text{Department of Biostatistics, Yale University}

\vspace{0.1in}

\text{Chuanhai Liu} \\
\text{Department of Statistics, Purdue University}

\vspace{0.1in}

\text{Heping Zhang} \\
\text{Department of Biostatistics, Yale University}




\end{center}
  \medskip
\tableofcontents
\clearpage

\subsection{Confidence Distribution from Perspective of the Imprecise Probability} \label{sp:imprecise}
In this section, we introduce the general definition of the confidence distribution from the imprecise probability perspective, elaborated in \cite{martin2021imprecise, martin2023b}. Following the notation in
\cite{martin2015plausibility, martin2018}, we define the \textit{contour} function, as a function of the observed data $y$ and the parameter value $\theta$:
\[
\pi_{y}(\theta)=\mathrm{P}_{Y \mid \theta}\left(T_{Y, \theta} \leq T_{y, \theta}\right), \quad \theta \in \Theta,
\]
which is exactly $F_\theta(T_{y,\theta})$ defined in Section 2.2. This function determines a possibility distribution defined by
\[
\bar{\Pi}_{y}(A)=\sup _{\theta \in A} \pi_{y}(\theta), \quad A \subseteq \Theta.
\]
This possibility distribution is an \textit{upper probability} and, therefore, corresponds to the upper envelope defined by a collection of ordinary probabilities, called the \textit{credal set} of $\bar{\Pi}_{y}$ :
\[
\mathscr{C}\left(\bar{\Pi}_{y}\right)=\left\{\mathrm{Q}_{y} \in \operatorname{probs}(\Theta): \mathrm{Q}_{y}(\cdot) \leq \bar{\Pi}_{y}(\cdot)\right\},
\]
which is the set of data-dependent probability distributions $\mathrm{Q}_{y}$ supported on $\Theta$ that are dominated by $\bar{\Pi}_{y}$. For possibility distributions, the credal set has a nice characterization such that:
\[
\mathrm{Q}_{y} \in \mathscr{C}\left(\bar{\Pi}_{y}\right) \Longleftrightarrow \mathrm{Q}_{y}\left(\left\{\theta: \pi_{y}(\theta) \leq \alpha\right\}\right) \leq \alpha, \quad \text { for all } \alpha \in[0,1].
\]
Since the set $\left\{\theta: \pi_{y}(\theta)>\alpha\right\}$ is a confidence set by construction (see Section 2.2), one can describe the credal set as the set of all those data-dependent distributions that assign at least probability $1-\alpha$ to all the $100(1-\alpha) \%$ confidence sets. This suggests defining the \textit{confidence distribution} $\mathrm{Q}_{y}^{\star}$ ({\it see}, \citeauthor{martin2021imprecise},\citeyear{martin2021imprecise, martin2023b}) as that which gives
\[
\mathrm{Q}_{y}^{\star}\left(\left\{\theta: \pi_{y}(\theta) \leq \alpha\right\}\right)=\alpha, \quad \text { for all } \alpha \in[0,1],
\]
if it exists. Note that the definition above is exactly the same as (4) in Definition~2.

\subsection{Technical Details of Theorem~3} \label{sps:proof-1}
\subsubsection{Regularity Conditions}

Firstly, we specify the regularity conditions for the likelihood function:
\begin{enumerate}[{C}1]
    \item \textbf{Differentiability}: For almost all $y$, $\frac{\partial}{\partial \theta} L_y(\theta)$ exists and is continuous for $\forall \theta \in \Theta$, where \( L_y(\theta) \) is the likelihood function.
    \item \textbf{Existence of an Interior Maximum}: There must be a point in the parameter space where the likelihood function attains its maximum. That is, 
    \[ \exists \, \hat{\theta} \in \Theta : L_y(\hat{\theta}) \geq L_y(\theta) \; \text{for all} \; \theta \in \Theta \]
\end{enumerate}
The key references for the proof are \cite{robbins1951stochastic} and \cite{dupa1982stochastic} which investigated the stochastic approximation algorithm. 

\subsubsection{Proof of Theorem~3}
The stochastic approximation algorithm aims at finding a unique root $m = m_\alpha$ such that $f^\alpha(m) \coloneqq \mbox{Prob}\left(F_{\hat{\theta}_*(m)}(T_{y, \hat{\theta}_*(m)})\le\alpha\right) \approx \alpha$.
First, we discuss the condition for the existence of the root. Recall that
\begin{align*}
    F_{\hat{\theta}_*(m)}(T_{y, \hat{\theta}_*(m)}) &= \mbox{Prob} \left(T_{Y, \hat{\theta}_*(m)} \leq T_{y, \hat{\theta}_*(m)}\right)\\
    &= \mbox{Prob} \left(\ell(Y, \hat{\theta}_Y) - \ell(Y, \hat{\theta}_*(m)) \leq \ell(y, \hat{\theta}_y) - \ell(y, \hat{\theta}_*(m)) \right).
\end{align*}
For finite $n$, as $m \to \infty$, $\hat{\theta}_*(m) \to \hat{\theta}_y$, $ F_{\hat{\theta}_*(m)}(T_{y, \hat{\theta}_*(m)}) \to 1$ and as a result, $f^\alpha(m) \to 0$ for any $\alpha$ ($0 \leq \alpha \leq 1$). In other words, for any $\epsilon$ such that $\alpha > \epsilon > 0$, $\exists M_0 > 0$ such that for $m > M_0$, $f^\alpha(m) < \epsilon$ and $\alpha \geq |f^\alpha(m) - \alpha| > \alpha - \epsilon$. 
Given the boundness of the function, the solution $m_\alpha = \min_{m \geq 1} |f^\alpha(m) - \alpha|$ must exist for any given $\alpha$.

To investigate the stochastic approximation process of $m$, some notations are introduced here. Let $m^{(t)}$ denote the updated $m$ value (an integer) at the $t$-th iteration. Denote by $\tilde{m}^{(t)} \in \mathbb{R}$ the obtained $m$ values by directly applying the stochastic update formula. Let $\tilde{m}^{(t)}_l$ and $\tilde{m}^{(t)}_r$ denote the left and right neighbouring integer values of $\tilde{m}^{(t)}$. The stochastic update steps can be expressed as follows:
\begin{enumerate}
    \item Sample $u$ from a uniform distribution: $u \sim \mbox{Uniform}(0,1)$
    \item Update $m^{(t)}$ by\\
    $m^{(t)} = \begin{cases}
    M_u & \text{if } \tilde{m}^{(t)} \geq M_u\\
     \tilde{m}^{(t)}_l \mathbb{I}\{u \geq \tilde{m}^{(t)} - \tilde{m}^{(t)}_l\} + \tilde{m}^{(t)}_r \mathbb{I}\{u \leq \tilde{m}^{(t)} - \tilde{m}^{(t)}_l\} & \text{if } M_u \geq \tilde{m}^{(t)} \geq M_l\\
    M_l & \text{if } \tilde{m}^{(t)} < M_l.
\end{cases}$
    \item Obtain an approximation $\widehat{f^{\alpha}(m^{(t)})}$ via Monte-Carlo approximation.
    \item Update $m$ for the next iteration as follows: \[m^{(t + 1)} = m^{(t)} + g_t \widehat{f^{\alpha}(m^{(t)})}.\]
\end{enumerate}
It is important to note that $M_u$ and $M_l$ are predefined positive integer values that serve as constraints to limit the range of $m$ during the update procedure, ensuring stable convergence. It is worth noting that this choice of step size satisfies the step size conditions outlined in \cite{robbins1951stochastic}, which are as follows:
\[
\sum_{t = 1}^\infty g_t = \infty,
\]
and
\[
\sum_{t = 1}^\infty g_t^2 < \infty.
\]
We define the step size $g_t$ as $C/t$ with $C > 0$. To align with the notation of \cite{dupa1982stochastic}, denote 
\[
f^{\alpha}(m^{(t)}) = \widehat{f^{\alpha}(m^{(t)})} + \mbox{err}(m^{(t)})
\]
where $\mbox{err}(m^{(t)})$ is the error of the Monte-Carlo approximation. As is in many other cases, the Monte-Carlo method here is unbiased. Thus, $E\left[ \mbox{err}(m^{(t)})\right]  = 0$.  Moreover, we have $E\left[ f^{\alpha}(m^{(t)})\right]^2 < \infty$ and $E\left[ \mbox{err}(m^{(t)})\right]^2 < \infty$, since $0 \leq f^{\alpha}(m^{(t)}) \leq 1$ and $0 \leq  \widehat{f^{\alpha}(m^{(t)})} \leq 1$. Thus, we can find a $K > 0$ such that for any $m^{(t)}$,
\[
\left(f^{\alpha}(m^{(t)})\right)^2 + E\left[ \mbox{err}(m^{(t)})\right]^2 \leq K \cdot\left(1 + (m^{(t)})^2\right).
\]
Together with this necessary condition and all the arguments above, the summarized four-step stochastic updates of $m$ converge to the result $m_\alpha$ with probability one by applying the result of Theorem 1 and Corollary 2 in \cite{dupa1982stochastic}.

Finally, given that the problem is bootstrap $\epsilon$-calibratable, \textit{i.e.} there exists an $m$ value such that 
\[
        \alpha \leq \mathrm{P}_{\hat{\theta}_*(m)}\left(F_{\hat{\theta}_*(m)}(T_{y, \hat{\theta}_*(m)})\leq \alpha\right) \leq \alpha +  \epsilon,
\]
or equivalently,
\[
\alpha \leq f^\alpha(m) \leq \alpha +  \epsilon,
\]
 since $\lfloor m^{(t)}\rfloor - 1$ is returned in the last step, this gives the desired solution  
\[ m_\alpha^* = \min_{m \geq 1} \{  f^\alpha(m) - \alpha : f^\alpha(m) \geq \alpha \}, \]
which yields confidence interval of at least $1-\alpha$ obtained with (13).

\begin{remark}
    Although the validity of the constructed confidence interval is established when the problem is bootstrap $\epsilon$-calibratable, the efficiency of the confidence set—how closely it approximates the target level—is dependent on $\epsilon$. The smaller the $\epsilon$ is, the closer to the target level. A smaller $\epsilon$ generally results in an interval that is closer to the target level. This dependency is problem-specific and fundamentally linked to the grid density of $f^\alpha(m)$ as $m$ varies. However, empirical examples, such as those presented in Example~2 and Section~S.4, demonstrate that the approximation is remarkably accurate. 
\end{remark}

\subsubsection{Additiona Illustrations of the Effectiveness of the Approximation} \label{sps:calibratable}

Here we conduct an additional investigation on the bound of the approximation error $|\alpha - f^\alpha(m_\alpha^*)|$ to gain deeper insights into the property of being bootstrap $\epsilon$-calibratable. 

Inspired by \cite{cella2024}, we approximate the contour function (or the confidence set using our term) $\{\theta: F_\theta (T_{y, \theta}) \geq \alpha \}$ using $\{\theta: (\theta - \hat{\theta})^T \Sigma^{-1}(\theta - \hat{\theta}) \leq \frac{\chi^2_{1-\alpha, p}}{c}\}$, which represents the $1-\alpha$ probability contour of a Gaussian distribution $N(\hat{\theta}, \frac{1}{c}\Sigma)$ with an unknown $\Sigma$ and a constant $c$. Meanwhile, we approximate the distribution of $\hat{\theta}_*(m)$ as a function of $m$ using the Gaussian distribution $N(\hat{\theta}, \frac{n}{m} \Sigma)$. This leads to
\begin{align*}
    f^\alpha(m) &= 1- \mathrm{P}_{\theta\sim N(\hat{\theta}, \Sigma)} \left(\frac{m}{n} (\theta - \hat{\theta})^T \Sigma^{-1}(\theta - \hat{\theta}) \leq \frac{\chi^2_{1-\alpha, p}}{c}\right) \\
    &= 1- \mathrm{P} \left(\chi^2_p > \frac{n}{m}\frac{\chi^2_{1-\alpha, p}}{c}\right).
\end{align*}
Note that we have as $m \to \infty$, $f^\alpha(m) \to 0$ as desired, and the approximate solution falls at $m_\alpha^* = \lfloor \frac{n}{c} \rfloor$. Furthermore, the condition for being bootstrap $\epsilon$-calibratable if there exists an $m > 0$ such that the model estimation process is feasible and $\frac{n}{m c} \geq 1$ (which yields $f^\alpha(m) \geq \alpha$). Additionally,  since $f^\alpha(m)$ is a decreasing function of $m$, the change in $f^\alpha(m)$ when increasing $m$ to $m + 1$ can be bounded by
\begin{align*}
\Delta f^\alpha(m) 
&\leq \chi^2_p\left(\frac{n}{m}\frac{\chi^2_{1-\alpha, p}}{c} \right) \cdot \left(\frac{n}{m}\frac{\chi^2_{1-\alpha, p}}{c} - \frac{n}{m + 1}\frac{\chi^2_{1-\alpha, p}}{c} \right) \\
&= \chi^2_p\left(\frac{n}{m}\frac{\chi^2_{1-\alpha, p}}{c} \right) \cdot \left(\frac{n}{m (m + 1)}\frac{\chi^2_{1-\alpha, p}}{c} \right) \\
&\leq \chi^2_p\left(\frac{m^*_\alpha}{m}\chi^2_{1-\alpha, p} \right) \cdot \left(\frac{m^*_\alpha}{m (m + 1)}\chi^2_{1-\alpha, p} \right),
\end{align*}
where $\chi^2_p(\cdot)$ is the density function of the $\chi^2$-distribution with $p$ degrees of freedom. By setting $m \approx m^*_\alpha$, we have
\[
\Delta f^\alpha(m_\alpha^*) \leq \chi^2_p\left(\chi^2_{1-\alpha, p}\right) \cdot \left(\frac{1}{m_\alpha^*}\chi^2_{1-\alpha, p}\right),
\]
which provides an approximate upper bound to the approximation error. 

In the finite sample case, consider $n = 20, p = 10$ as an example. In this case, the upper bound for the approximation error is approximately 0.01, which is sufficiently small for practical applications. Asymptotically, this analysis implies that the approximation error can be negligible when $m$ and $n$ are large, and it vanishes as $n\to \infty$ and $m \to \infty$. The Gaussian approximation used in this analysis is also asymptotically validated by Wilks' theorem \citep{wilk1938}.

\subsection{Technical Details of Corollary~1}\label{sps:proof-2}
\subsubsection{Proof of Corollary~1}
Condition C1 suggests that for $\theta_1$ and $\theta_2$ such that $\ell(y, \theta_1) \geq \ell(y, \theta_2)$, we have $F_{\theta_1}(T_{y, \theta_1}) \leq F_{\theta_2}(T_{y, \theta_2})$. The property of cumulative distribution functions (CDFs) suggests that $F_{\theta}(T_{y,\theta}) \in [0,1]$, and the continuity of $T_{y,\theta}$ suggests that it has non-zero density at any point in $[0,1]$. Taking the regularity conditions of the loss function into consideration, for any $t \in \mathbb{R}$, we can find $t' \in \mathbb{R}$ such that
\[
\mbox{Prob}\left(\ell(y, \theta) \leq t\right) = \mbox{Prob}\left( F_{\theta}(T_{y,\theta}) \geq t'\right).
\]
Under C2, for $m_1 > m_2 > 0$, if $\ell(y, \hat{\theta}_*(m_2))$ first-order stochastically dominates $\ell(y, \hat{\theta}_{*}(m_1))$, for any $t \in \mathbb{R}$, we have
\[
\mbox{Prob}\left(\ell(y, \hat{\theta}_*(m_1)) \leq t \right) \geq \mbox{Prob}\left(\ell(y, \hat{\theta}_*(m_2)) \leq t \right).
\]
Taking the above observation into consideration, we have
\[
\mbox{Prob}\left(F_{\hat{\theta}_*(m_1)}(T_{y, \hat{\theta}_*(m_1)})\geq t' \right) \geq \mbox{Prob}\left(F_{\hat{\theta}_*(m_2)}(T_{y, \hat{\theta}_*(m_2)}) \geq t' \right)
\]
for any $t' \in \mathbb{R}$. This  suggests
\[
\mbox{Prob}\left(F_{\hat{\theta}_*(m_1)}(T_{y, \hat{\theta}_*(m_1)})\leq t' \right) \leq \mbox{Prob}\left(F_{\hat{\theta}_*(m_2)}(T_{y, \hat{\theta}_*(m_2)}) \leq t' \right).
\]
Taking $t' := \alpha$ gives the desired results of $f^\alpha(m_1) \leq f^\alpha(m_2)$.

Next, we prove the results again by replacing condition C1 with its sufficient condition, wherein the distribution of $T_{Y,\theta}$ remains invariant under $\theta$. Now, we can write
\begin{align*}
    f^\alpha(m) &\coloneqq \mbox{Prob}\left(F_{\theta}(T_{y, \hat{\theta}_*(m)})\le\alpha\right) \\
    &\coloneqq \mbox{Prob}\left(F_{\theta}\left(\ell(y, \hat{\theta}_y) - \ell(y, \hat{\theta}_*(m))\right)\le\alpha\right)
\end{align*}
where 
\begin{equation}\label{eq:proof-1}
    F_\theta(t) = \mbox{Prob}\left(T_{Y,\theta} \leq t \right)
\end{equation}
is the CDF of $T_{Y,\theta}$. Under C2, for $m_1 > m_2 > 0$, if $\ell(y, \hat{\theta}_*(m_2))$ first-order stochastically dominates $\ell(y, \hat{\theta}_{*}(m_1))$, 
\[
\mbox{Prob}\left(\ell(y, \hat{\theta}_*(m_1)) \leq t \right) \geq \mbox{Prob}\left(\ell(y, \hat{\theta}_*(m_2)) \leq t \right)
\]
where $\hat{\theta}_*(m_1) \sim \mathrm{G}_{y,m_1}$ and $\hat{\theta}_*(m_2) \sim \mathrm{G}_{y,m_2}$, for any $t \geq \ell(y, \hat{\theta}_y)$. To prove $f^\alpha(m_1) \leq f^\alpha(m_2)$ by contradiction, assume that there is a $\alpha$ such that $f^\alpha(m_1) > f^\alpha(m_2)$. This suggests that
\[
\mbox{Prob}\left(\ell(y, \hat{\theta}_y) - \ell(y, \hat{\theta}_*(m_1)) \leq F^{-1}(\alpha) \right) \geq \mbox{Prob}\left(\ell(y, \hat{\theta}_y) - \ell(y, \hat{\theta}_*(m_2)) \leq F^{-1}(\alpha) \right)
\]
and one can further write
\[
\mbox{Prob}\left(\ell(y, \hat{\theta}_*(m_1)) \leq \ell(y, \hat{\theta}_y) - F^{-1}(\alpha) \right) \leq\mbox{Prob}\left(\ell(y, \hat{\theta}_*(m_2)) \leq \ell(y, \hat{\theta}_y) - F^{-1}(\alpha) \right)
\]
which contradicts (\ref{eq:proof-1}) by taking $t = \ell(y, \hat{\theta}_y) - F^{-1}(\alpha)$. Note that $F^{-1}(\alpha) \leq 0$ any $1 \geq \alpha \geq 0$.

Finally, we prove Theorem~3 holds true asymptotically while $n \to \infty$. First note that $T_{Y, \theta} = \ell(Y, \hat{\theta}_Y) - \ell(Y, \theta)$ takes the form of likelihood ratio. Suppose that the regularity conditions of the likelihood function for establishing the asymptotic normality of MLE hold (see \cite{vaart1998}). By Wilks' Theorem on likelihood ratio \citep{wilk1938}, we have $-2 T_{Y,\theta}$ converges asymptotically to a chi-square distribution with $Y\sim \mathrm{P}_\theta$. This further suggests that the dependence of $F_\theta$ on $\theta$ disappears as $n \to \infty$. Thus, the sufficient condition for C1 holds true asymptotically.

For C2, notes that the asymptotic variance of MLE under regularity conditions satisfy
\[
\sqrt{n} (\hat{\theta}_n - \theta) \xrightarrow{d} N (0, I(\theta)^{-1})
\]
where $I(\theta)$ is the Fisher Information of $\theta$.  We can see the empirical distribution of $y$ as some true population distribution with parameter $\hat{\theta}_y$, where $\hat{\theta}_y$ is the MLE of $\theta$ given $y$. In this case, $m$-out-of-$n$ bootstrap can be seen as a sample with size $m$ from this true distribution. This implies
\[
\sqrt{m} (\hat{\theta}_*(m) - \hat{\theta}_y) \xrightarrow{d} N (0, I(\hat{\theta}_y)^{-1})
\]
as $m \to \infty$, which suggests that the asymptotic distribution of $\hat{\theta}_*(m)$ is a Gaussian distribution with decreasing variance as $m_1 > m_2$. This implies the asymptotic first-order stochastic dominance of $\ell(y, \hat{\theta}_*(m_2))$ over $\ell(y, \hat{\theta}_{*}(m_1))$ under mild conditions of the likelihood function. 

Thus, both conditions hold true asymptotically. 

\subsubsection{Examples}\label{sps:lr-condition-example}
Here we show that the model used in Example~2 satisfies the conditions mentioned in Corollary~2. For the linear regression model, we notice that 
\begin{align*}
    T_{y,\theta} &= \ell(y , \hat{\theta}_y) - \ell(y, \theta)\\
    &= \frac{n}{2}\log (2\pi \hat{\sigma}^2) + \frac{1}{2\hat{\sigma}^2}(y - X\hat{\beta}_y)^T(y - X \hat{\beta}_y) -\frac{n}{2}\log (2\pi \sigma^2) - \frac{1}{2{\sigma}^2}(y - X\beta)^T(y - X\beta)\\
    &= \log (\frac{\hat{\sigma}^2}{\sigma^2}) + \frac{1}{2\hat{\sigma}^2}(y - X\hat{\beta}_y)^T(y - X \hat{\beta}_y) - \frac{1}{2{\sigma^2}}(y - X\beta)^T(y - X\beta).
\end{align*}
Given that 
\[
\frac{\hat{\sigma}^2}{\sigma^2} \sim \frac{1}{n} \chi^2_{n-2}
\]
\[
\frac{1}{2\hat{\sigma}^2}(y - X\hat{\beta}_y)^T(y - X \hat{\beta}_y) = \frac{n}{2}
\]
\[
\frac{1}{2\hat{\sigma}^2}(y - X\hat{\beta}_y)^T(y - X \hat{\beta}_y) - \frac{1}{2{\sigma^2}}(y - X\beta)^T(y - X\beta) \sim \frac{1}{2}\chi^2_{n-p},
\]
the function $T_{y,\theta}$ does not depend on $\theta$, which implies condition C1 to hold true. In Example 2, $\sigma$ is assumed to be given. In such case, we have
\begin{align*}
    T_{y,\theta} &= \ell(y , \hat{\theta}_y) - \ell(y, \theta)\\
    &= \frac{1}{2}(y - X\hat{\beta}_y)^T(y - X \hat{\beta}_y) - \frac{1}{2}(y - X\beta)^T(y - X\beta) \\
    &= \frac{1}{2}(y - X\hat{\beta}_y)^T(y - X \hat{\beta}_y) - \frac{1}{2}(y - X \hat{\beta}_y + X \hat{\beta}_y -X\beta)^T(y - X \hat{\beta}_y + X \hat{\beta}_y - X\beta)\\
    &= \underbrace{(y - X\hat{\beta}_y)^TX}_{=0}(\hat{\beta}_y -\beta) - \frac{1}{2}(\hat{\beta}_y -\beta)^TX^TX (\hat{\beta}_y - \beta)\\
    &= - \frac{1}{2}(\hat{\beta}_y -\beta)^TX^TX (\hat{\beta}_y - \beta) \\
    &\sim \frac{1}{2}\sigma^2 \chi^2_p,
\end{align*}
which also does not depend on $\beta$.

For condition C2, the stochastic dominance of the loss function with regard to the decreasing $m$ is closely related to the concept of \textit{risk monotonicity} \citep{Loog2019} in the learning theory: with the increasing sample size, the expectation of the loss monotonically decreases. Although desirable, the rigorous proof of such property for a general case is considered to be difficult, and the existing research on the property generally rely on specific examples \citep{Loog2019}. Here we provide an empirical evidence of such property for Example 2, with the case $n = 100, \kappa = 0.3$, shown in Figure~\ref{sps:proof-2:fig:1}.
\begin{figure}
\centering
\includegraphics[width=5.5in]{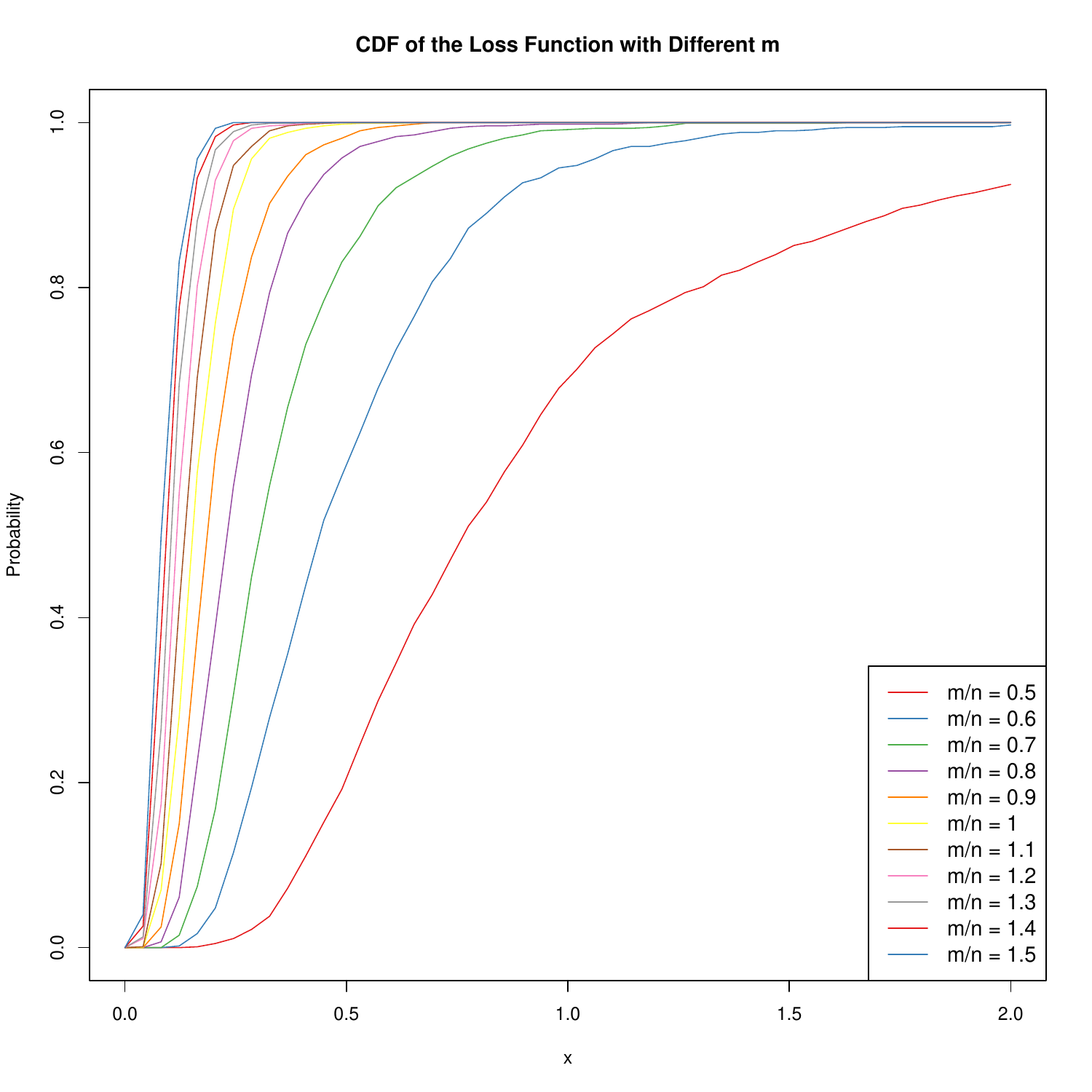}
        \caption{The CDF of the shifted loss function value $\ell(y, \hat{\theta}_*(m)) - \ell(y, \hat{\theta}_y)$, with different values of $m$, using the model in Example~2 with a simulated data set. The estimation is based on 1,000 Monte-Carlo repetitions.
        }
\label{sps:proof-2:fig:1}
\end{figure}

\subsection{An Empirical Study of the RA algorithm for Approximating the Confidence Distribution}\label{sp:fr-counter-example}
Consider the example used in Sections~2.2 and 3.2. Suppose we are interested in the inference of the mean parameter $\theta$ with $y$, a sample containing $n$ observations from the model $Y \sim N(\theta, 1)$. Suppose now the estimation is conducted with the penalty value $\lambda$. That is, the maximized penalized negative log-likelihood estimator is $\hat{\theta} = \mbox{sign}\left( \hat{\theta}_y \right)
\left(|\hat{\theta}_y| - \lambda\right)_+,$ where $\hat{\theta}_y$ is the MLE of $\theta$ and
\begin{equation*}
\mbox{sign}\left(z\right)\left(|z| - \gamma\right)_+ =
\left\{
\begin{array}{lcl}
z - \gamma && \mbox{if 
$z > 0$ and $\gamma < |z|$}\\
z + \gamma && \mbox{if 
$z < 0$ and $\gamma < |z|$}\\
0 && \mbox{if 
$\gamma \geq |z|$.}\\
\end{array}
\right.
\end{equation*}
For RA, the loss function used in the generalized association function now becomes the penalized loss function
\[
\ell(y, \theta) = \frac{1}{2}\sum_{i = 1}^n(y_i - \theta)^2 + \lambda |\theta|.
\]

Here we consider the case with $\theta = 1, n = 100,$ and $\lambda = 0.2$. For constructing the confidence distribution of $\theta$ denoted by ${\theta}_* \sim G_y$ with samples of $\hat{\theta}_* \in \Theta$, it is required that the distribution of $F_{\hat{\theta}_*}(T_{y, \hat{\theta}_*})$ follows the standard uniform distribution (see Definition~2 and Theorem~2). In the example, RA process following the refinement process by the DR refinement process
(Algorithm~2) is used for approximating such distribution.

A natural question is, can the standard bootstrap be directly used in the DR refinement process? If we bypass the RA process and directly use the standard bootstrap for generating candidate $\hat{\theta}_*$ values, the distribution of the obtained $F_{\hat{\theta}_*}(T_{y, \hat{\theta}_*})$ values (estimated with Monte-Carlo) is given in Figure~\ref{sp:fig:empirical-standard-bootstrap}. It is evident that the standard bootstrap fails to cover the tail values where $F_{\hat{\theta}_*}(T_{y, \hat{\theta}_*})$ is close to $0$. Consequently, even with the DR refinement process, the standard bootstrap does not provide a satisfactory approximation of the true confidence distribution due to the absence of these critical tail values.

We then conduct the RA process at various target significance levels, ranging from 0.05 to 0.95. For each of the target level, RA returns a $m$ value corresponding to the resampling size, which can be used for the $m$-out-of-$n$ bootstrap to conduct inference at the desired level. The resulting distribution of the $F_{\hat{\theta}_*}(T_{y, \hat{\theta}_*})$ across these $19$ different $m$ values are shown in Figure~\ref{sp:fig:empirical-fr-1}. Notably, the mean of the distribution shifts from $1$ to $0$ as the target significance level increases. More importantly, unlike the standard bootstrap, it consistently covers both tails.

Next, we study the mixture of the obtained distributions of $F_{\hat{\theta}_*}(T_{y, \hat{\theta}_*})$ with $m$ values corresponding to different significance levels. We consider two different mixtures one with levels spaced at regular intervals from 0.05 to 0.95, and another more sparse mixture at 0.05, 0.50, and 0.95. The distributions for these two cases, as shown in Figure~\ref{sp:fig:empirical-fr-2}, reveal no significant difference, suggesting a relatively minor impact of the grid density of target significance levels. However, despite the apparent uniformity within the central range from 0.1 to 0.9, both tails exhibit considerably higher density mass. Consequently, the DR refinement process is still necessary for ``reweighting'' the distribution, ensuring the overall uniformity.

\begin{figure}
\centering
\includegraphics[width=5.5in]{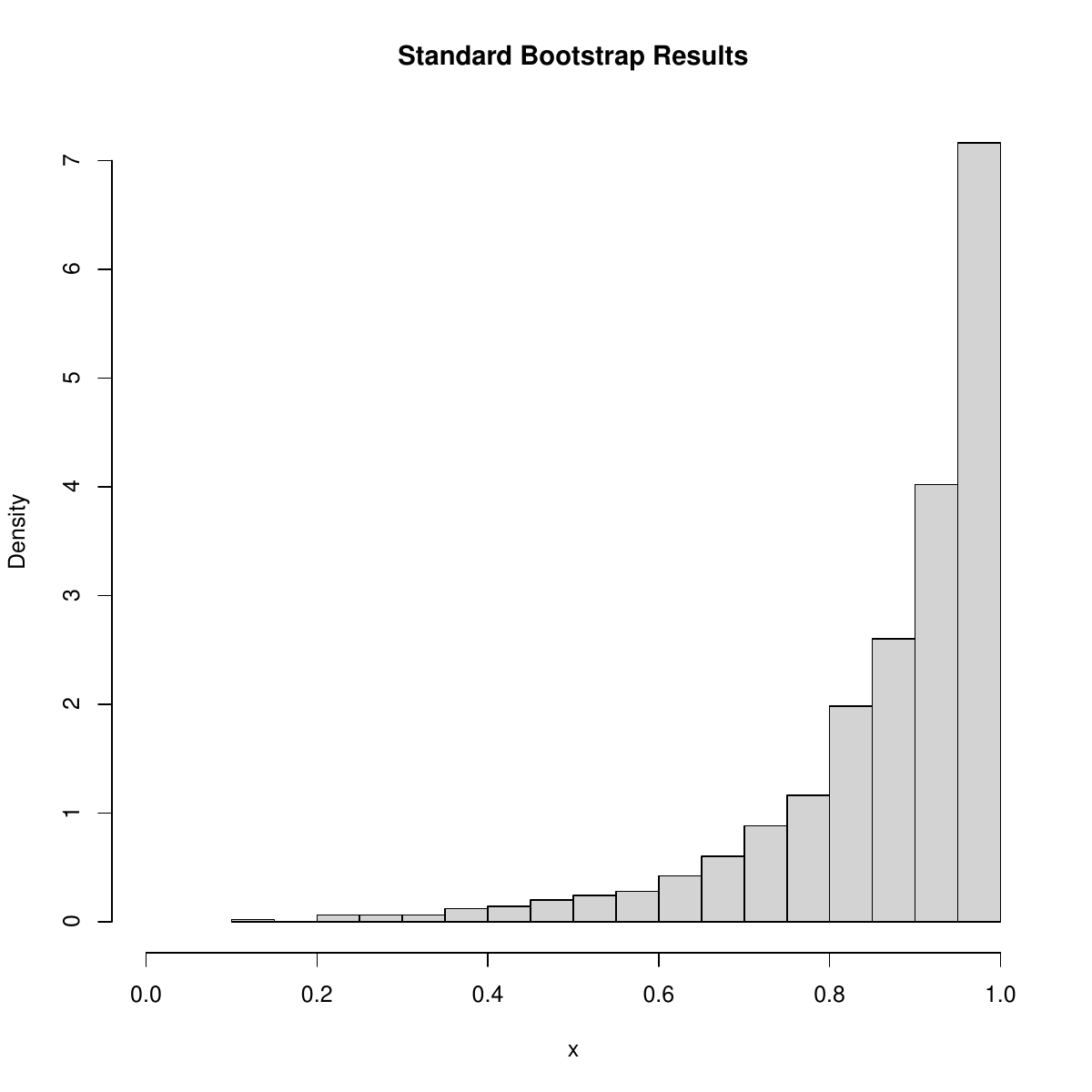}
        \caption{The distribution of $F_{\hat{\theta}_*}(T_{y, \hat{\theta}_*})$ obtained with the standard $n$-out-of-$n$ bootstrap (1,000 resampling repetitions).
        }
\label{sp:fig:empirical-standard-bootstrap}
\end{figure}

\begin{figure}
\centering
\includegraphics[width=5.5in]{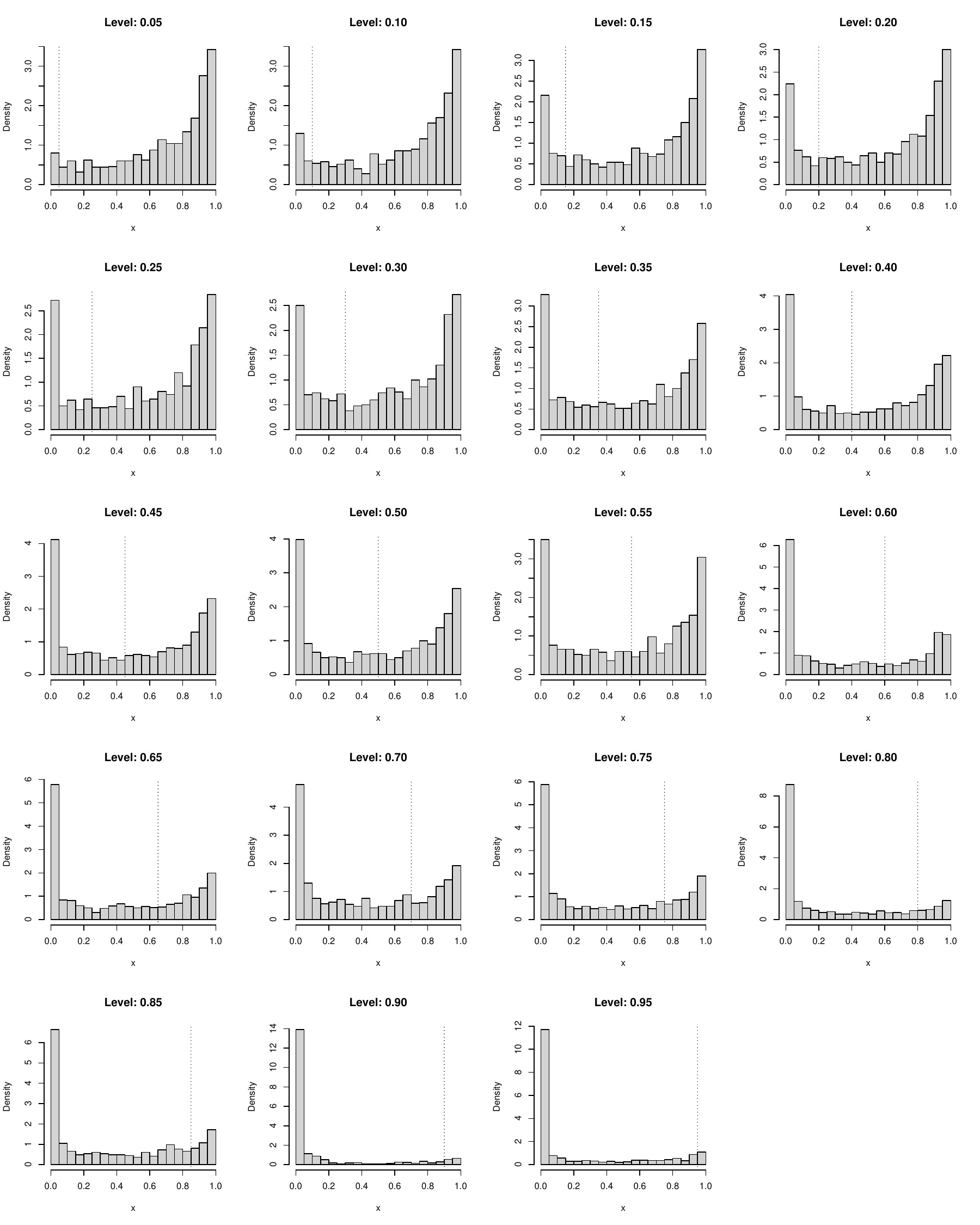}
        \caption{The distribution of $F_{\hat{\theta}_*}(T_{y, \hat{\theta}_*})$ obtained with varying $m$ values, which are output of the SA algorithm corresponding to different target level $\alpha$ (1,000 resampling repetitions). 
        }
\label{sp:fig:empirical-fr-1}
\end{figure}

\begin{figure}
\centering
\includegraphics[width=5.5in]{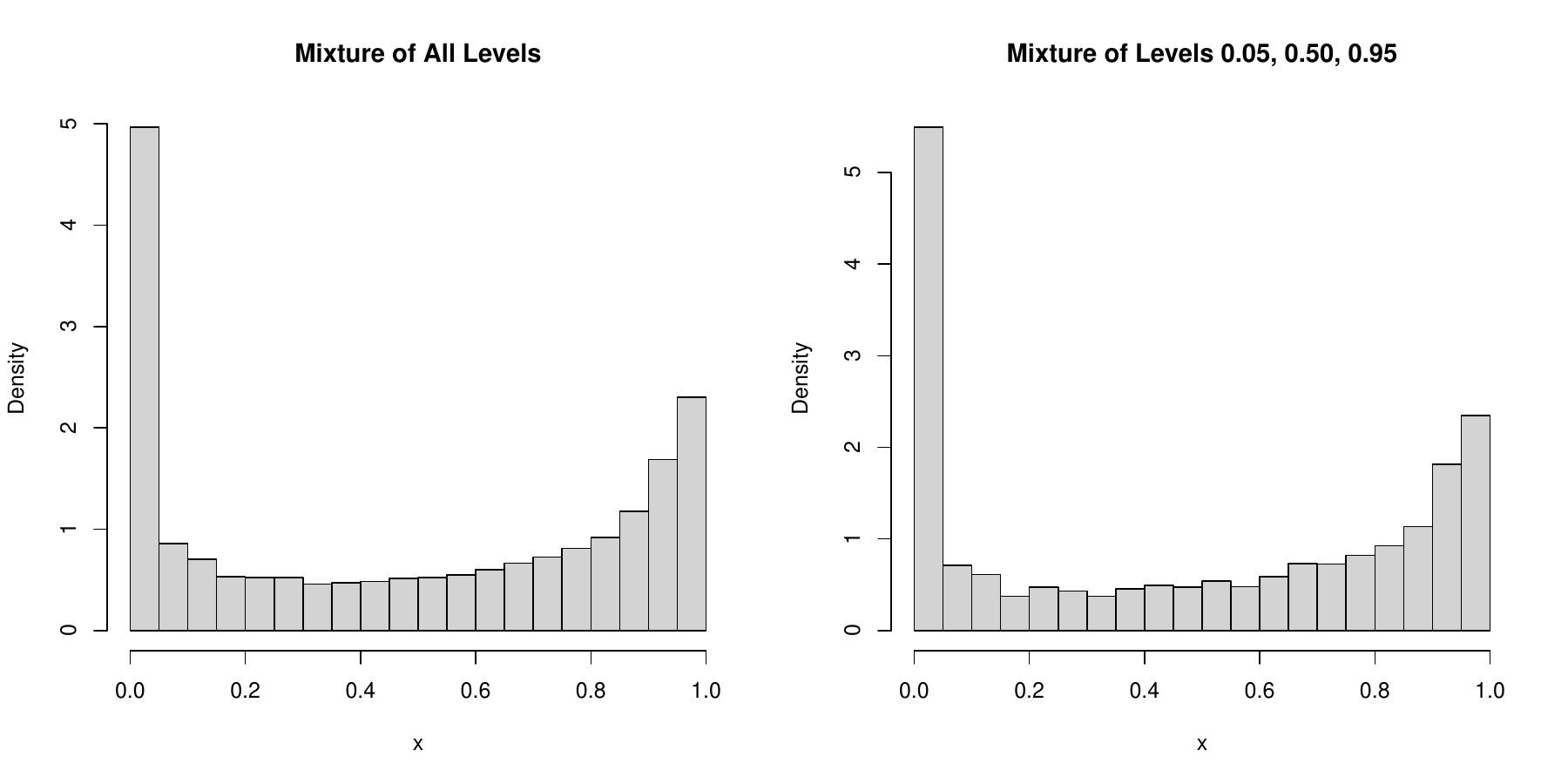}
        \caption{The mixture distribution of $F_{\hat{\theta}_*}(T_{y, \hat{\theta}_*})$ with two different grid density of target level $\alpha$. The left histogram corresponds to the target levels $(0.05, 0.10,\dots,0.95)$, and the right histogram corresponds to the target levels $(0.05,0.50,0.95)$.
        }
\label{sp:fig:empirical-fr-2}
\end{figure}

\subsection{Fiducial Distribution of $\beta$ in High-dimensional Example}\label{sps:fiducial-distribution-beta}
For the estimates
\[
\hat{\beta} = (X'X)^{-1}X' Y, \quad \hat{\sigma} =  \frac{1}{n-p}Y'\left(I- X(X'X)^{-1}X'\right)Y
\]
the association function which serves as the pivotal quantity can be written as
\[
    \frac{\hat{\beta} - \beta}{\sigma} = Z_1,\quad Z_1 \sim N_p\left(0, (X'X)^{-1}\right),
\]
where $N_p(\cdot)$ denotes the multivariate Gaussian distribution with dimension $p$ and
\[
    \frac{\hat{\sigma}}{\sigma} = Z_2,\quad (n - p)Z_2^2 \sim \chi^2\left(n -p\right),
\]
where $\chi^2(n -p)$ denotes chi-square distribution with $n -p$ degree of freedom. We are interested in a resampling scheme that can produce the resampled estimates of $\beta$ which is close to its so-called fiducial distribution
\[ 
    \beta= \hat{\beta} - \hat{\sigma}Z,\quad Z \sim t_p\left(0, (X'X)^{-1}, n-p\right).
\] 
where $t_p(\cdot)$ denotes the $p$-dimensional multivariate student-$t$ distribution. Note that this expression can be further written as
\[
\beta = (X'X)^{-1}X' \left(X\hat{\beta}+\hat{\sigma} \varepsilon\right) , \qquad \varepsilon \sim t_n(0, I, n-p). 
\]
We see that the best way to sample $\beta$ from its fiducial distribution is to generate the new observation $(X, Y_*)$ with $Y_* = X\hat{\beta}+\hat{\sigma}\varepsilon$ where $ \varepsilon \sim t_n(0, I, n-p)$, followed by the least-square estimate of $\beta$
\[
\hat{\beta}_* = (X'X)^{-1}X'Y_*.
\]
This way of creating new samples using parametric models is a type of parametric bootstrap \citep{efron2012bayesian}. The resulting distribution of $\hat{\beta}_*$ is given by
\[
    t_p\left(\hat{\beta}, \hat{\sigma}^2 (X'X)^{-1}, n-p\right), 
\]
and can be used to provide an exact confidence region of $\beta$. Similarly, when $\sigma$ is known,  the resulting distribution of $\hat{\beta}_*$ is given by
\[
    N_p\left(\hat{\beta}, \sigma^2 (X'X)^{-1}\right). 
\]

It is interesting to see that when $\sigma$ is assumed to be known, the best way to generate new samples becomes $Y_* = X\hat{\beta}+\sigma \varepsilon$ where $\varepsilon \sim \mbox{N}_n(0, I)$. This is exactly the standard parametric bootstrap \citep{efron2012bayesian}. More discussion about the parametric bootstrap is provided next.

\subsection{Parameteric Bootsrap in High-dimensional Example}\label{sps:param-bootstrap}
 When $\sigma$ is unknown, it can be seen that the standard parametric bootstrap in the high-dimensional linear regression example mentioned above can not provide an exact confidence region of $\beta$ in finite-sample case, due to the difference in the distribution of $\epsilon$ when generating $Y^*$. The question is, can standard parametric bootstrap lead to the valid inference while $n \to \infty$?  
{ We see the variability of the standardized $\hat{\beta}_*$ has \[
\frac{1}{s^2}(\hat{\beta}_*-\hat{\beta})'(X'X)(\hat{\beta}_*-\hat{\beta}) \sim \chi^2_p
\]
for the standard parametric bootstrap whereas for the Student-t parametric bootstrap
 \[
\frac{1}{s^2}(\hat{\beta}_*-\hat{\beta})'(X'X)(\hat{\beta}_*-\hat{\beta})\sim \frac{\chi^2_p}{\chi^2_{n-p}/(n-p)}.
\]}
In the asymptotic context with $p/n \rightarrow c >0 $ as $n\rightarrow \infty$, it is easy to see that
 with a slightly abused notation for independent chi-squared random variables $\chi_p^2$ and $ \chi_{n-p}^2$, we have
 \begin{eqnarray}
    \sqrt{p}\left[\frac{\chi_p^2/p}{\chi_{n-p}^2/(n-p)}-1\right] 
    &=&\sqrt{p}\left[\frac{1+\frac{1}{\sqrt{p}}\frac{\chi_p^2-p}{\sqrt{p}}}{1+\frac{1}{\sqrt{n-p}}\frac{\chi_{n-p}^2-(n-p)}{\sqrt{n-p}}} - 1 \right] \nonumber \\
    &=&\frac{\frac{\chi_p^2-p}{\sqrt{p}}-\sqrt{\frac{p}{n-p}}\frac{\chi_{n-p}^2-(n-p)}{\sqrt{n-p}}}{1+\frac{1}{\sqrt{n-p}}\frac{\chi_{n-p}^2-(n-p)}{\sqrt{n-p}}}\nonumber \\
    &\stackrel{D}{\rightarrow}& \mbox{N}(0, 1+c), \quad\ \mbox{ as } n\rightarrow \infty,\label{eq:LR-desired-Bootstrap}
\end{eqnarray}
where \eqref{eq:LR-desired-Bootstrap} follows Slutsky's theorem. In contrary, for the Gaussian parametric bootstrap, we have the corresponding result:
\begin{eqnarray}
    \sqrt{p}\left[\chi_p^2/p-1\right] 
    &=&\frac{\chi_p^2-p}{\sqrt{p}}
    \; \stackrel{D}{\rightarrow} \; \mbox{N}(0, 1), \quad\ \mbox{ as } n\rightarrow \infty.\label{eq:LR-parametric-Bootstrap}
\end{eqnarray}
The difference can be seen when comparing the asymptotic distribution of (\ref{eq:LR-desired-Bootstrap}) and (\ref{eq:LR-parametric-Bootstrap}). Thus, the gap between the two methods still exists as $n \to \infty$.

\subsection{Pairs Bootstrap and Residual Bootstrap in High-dimensional Example}\label{sps:residual-bootstrap}

In the linear regression context, an alternative approach to the standard pairs bootstrap known as the residual bootstrap \citep{Freedman1981} is also commonly used in  practice. Denote by $\hat{\beta}_{ols}$ the least square estimate of $\beta$. Define the residuals
\[
    e_i = y_i - x_i' \hat{\beta}_{ols}, i = 1, \dots, n
\]
and consider the set of centered residuals $\{e_1 - \bar{e}_n, \dots, e_n - \bar{e}_n \}$, where  $\bar{e}_n = \frac{1}{n}\sum_{i = 1}^n e_i$. The residual bootstrap selects a random sample $\{\tilde{e}_i\}_{i=1}^n$ of size $n$ with replacement from this set. The bootstrapped response $\tilde{y}_i$, $i = 1,\dots, n$ is then created as
\[
    \tilde{y}_i = x_i'\hat{\beta}_{ols} + \tilde{e}_i, i = 1,\dots, n.
\]

Residual bootstrap is considered to be less conservative compared to pairs bootstrap, as it does not involve resampling of the predictors \citep{el2018can}. With the observation in Example~2 that pairs bootstrap tend to generate more over-dispersed $\beta$ than the true distribution, residual bootstrap seems to be a plausible choice in such cases. To check this intuition, further simulation is performed for the same high-dimensional linear regression example with $\kappa = 0.3$. Different sample sizes ($n = 20,50,100,500$) are considered. 

\begin{figure}
\centering
\includegraphics[width=5.5in]{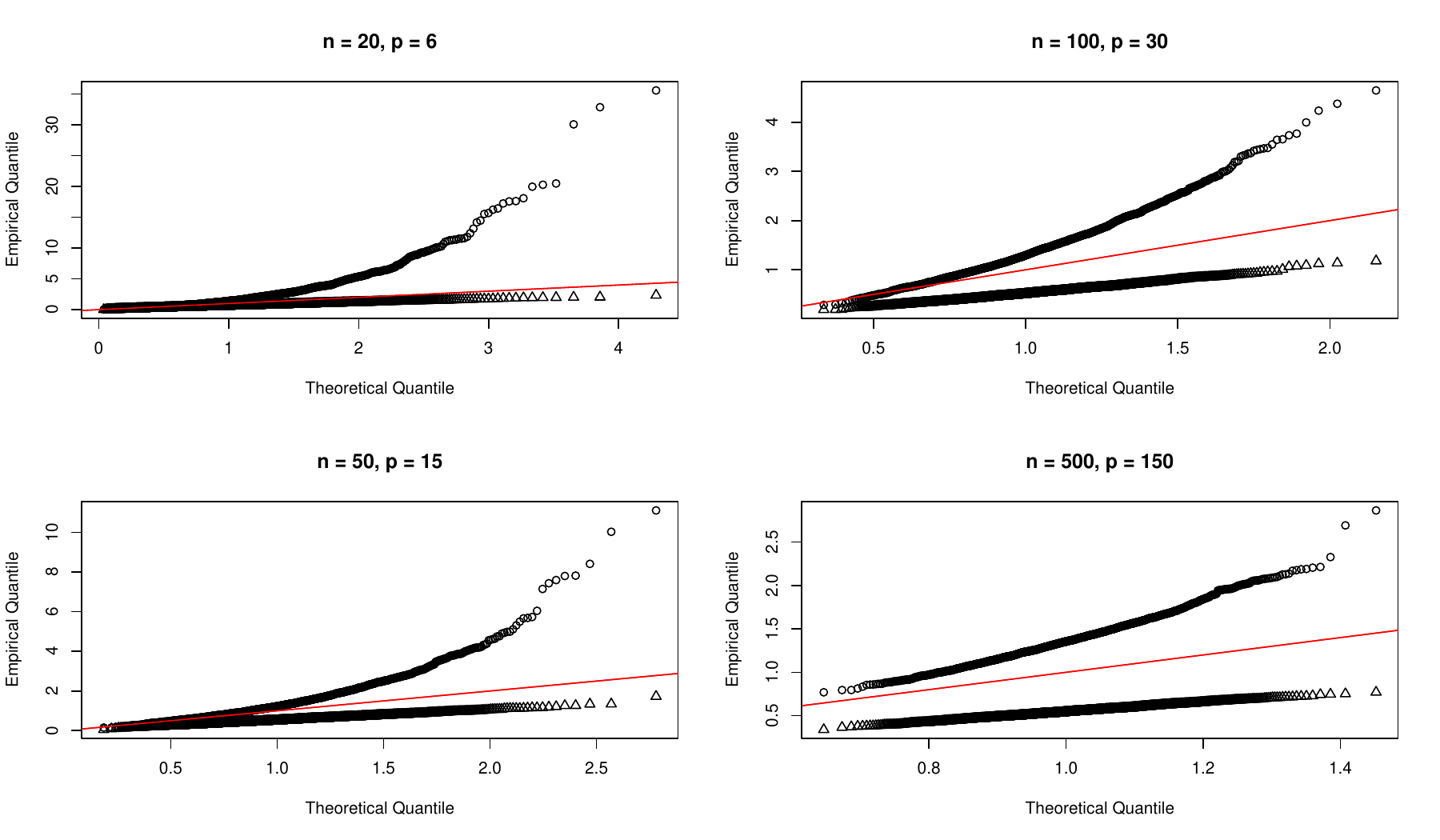}
        \caption{The Q-Q plot of 1,000 bootstrap estimates of $(\beta-\hat{\beta})'(X'X)(\beta-\hat{\beta})/p$ against the theoretical quantiles. The circle ``$\circ$'' points are obtained by the pairs bootstrap and the triangle ``$\triangle$'' points by the residual bootstrap. 
        }
\label{fig:residual-1}
\end{figure}

Figure~\ref{fig:residual-1} compares the distribution of standard pairs bootstrap approximations of $(\beta-\hat{\beta})'(X'X)(\beta-\hat{\beta})/p$ obtained with two different bootstrap methods against the desirable theoretical distribution. It can be seen that for all $n$ values, the pairs bootstrap tends to produce a more over-dispersed $\beta$ whereas the residual bootstrap, being a less conservative approach, tends to yield an under-dispersed $\beta$. This highlights the limitation of both the pairs bootstrap and the residual bootstrap, underscoring the need for a refined procedure such as CB proposed in this article.

\subsection{An Additional Example with DR Algorithm}\label{sp:fr-addtional-example}

Here we demonstrate the effectiveness of the DR algorithm (Algorithm~2) in approximating the true fiducial distribution with an example in \cite{martin2023b}.

Consider direction measurements data on the plane, which can be represented by angles relative to a reference point. The roulette wheel data from Example 1.1 in \cite{Mardia2009} is used here. The $n = 9$ observations of this data is obtained by spinning a roulette wheel and recording the angle of the position at which the wheel stops. Figure~\ref{sp:fig-von-data} provides a visualization of the data. Denote $y = (y_1, \dots, y_n)$ the observed angles of these $n$ independent spins $Y_1,\dots,Y_n$.  We consider model the data with the \textit{von Mises distribution}, which has the density function
\[
p(\theta) = \frac{1}{2\pi I_0(\kappa)} \exp\{\kappa \cos(y - \theta)\}, \quad y, \theta \in [0, 2\pi),
\]
where $\kappa > 0$ is a known concentration parameter, and $I_0$ denotes the Bessel function of order 0. We assume here $\kappa = 2$ is known. The minimal sufficient statistics of this model is
\[
X = (\bar{C}, \bar{S}) = \left( \frac{1}{n} \sum_{i=1}^n \cos Y_i, \frac{1}{n} \sum_{i=1}^n \sin Y_i \right).
\]
Convert this average position to polar coordinates, we have
\[
G = \arctan\left(\bar{S}/\bar{C}\right) \quad \text{and} \quad U = \sqrt{\bar{C}^2 + \bar{S}^2}.
\]
Given $U = u$, it is then shown in \cite{martin2023b} that the (conditional) fiducial density takes the form  
\[
q_x(\theta) \propto \exp\{\kappa u \cos(g - \theta)\}, \quad \theta \in (0, 2\pi].
\]

\begin{figure}[htbp]
\centering
\includegraphics[width=5.5in]{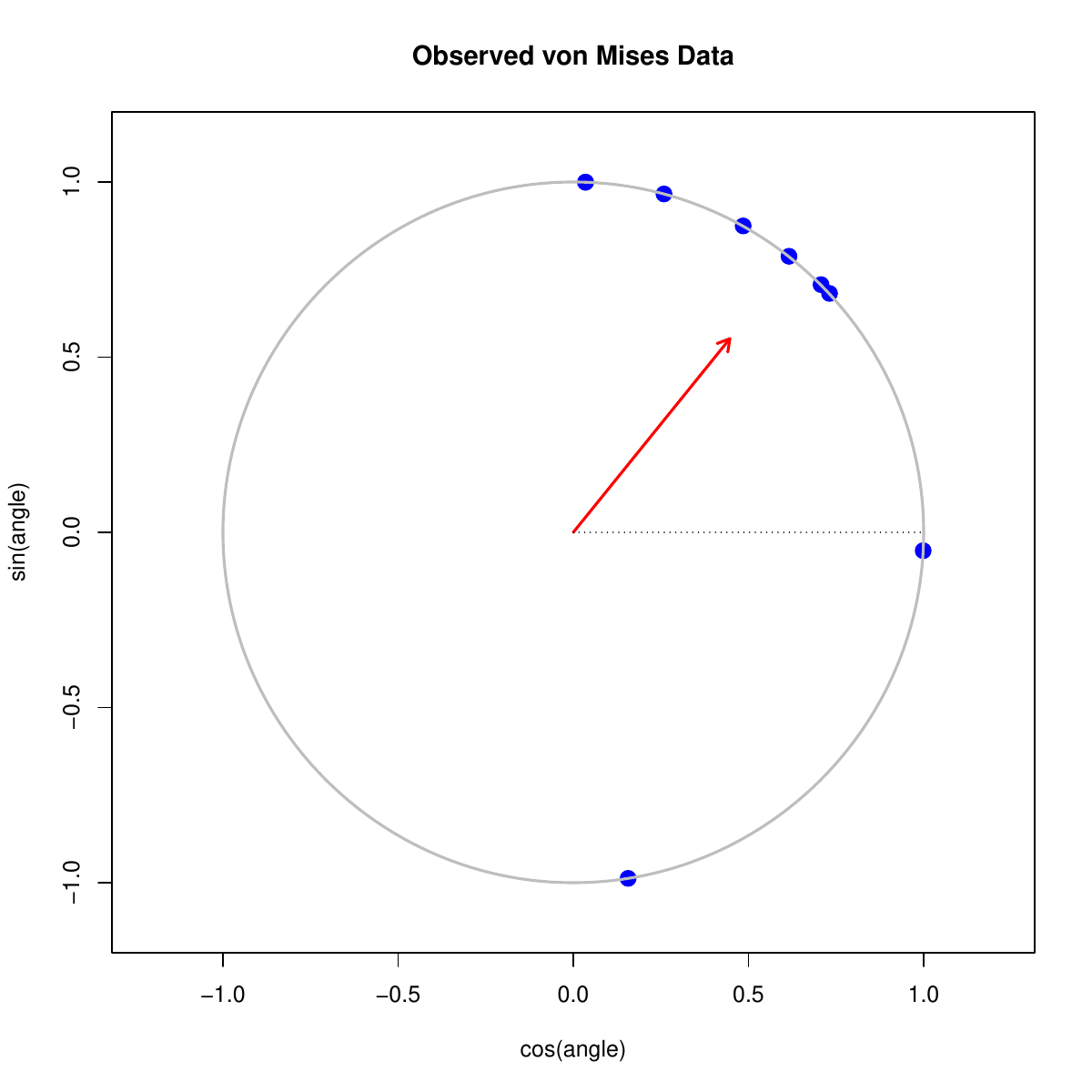}
        \caption{Visualization of the observed direction measurements data in \cite{Mardia2009}. The solid circles represent the Cartesian coordinates corresponding to the observed data. The circle that is closest to $(0,1)$ represent two observations taking exactly the same value. The arrow is pointing to the mean of the Cartesian coordinates of all observed data. The angle of the arrow makes with the dotted reference line is the maximum likelihood estimator $G$ of $\Theta$. The length of the arrow is the sample concentration $U$. \label{sp:fig-von-data}
        }
\end{figure}

To implement the proposed DR algorithm, we first take a grid of candidate $\theta$ values within the interval $[0 ,2\pi)$ and compute the function value $F_{\theta}(T_{y,\theta})$ via Monte-Carlo (100 repetitions). The resulting values as the function of $\theta$ are displayed in Figure~\ref{sp:fig-von-fr} (a). The DR algorithm then generates an approximated fiducial distribution, as illustrated in Figure~\ref{sp:fig-von-fr} (b). It is shown that the algorithm can provide a very close approximation to the true fiducial distribution, demonstrating the algorithm’s efficacy. 

Given that the true fiducial distribution is known, a numerical comparison between the empirical Cumulative Distribution Function (CDF) of the approximated distribution and the theoretical CDF is feasible. The results of this comparison are presented in Table~\ref{sp:tab:cdf_comparison}. It demonstrates that the proposed approximation algorithm has very high accuracy.


\begin{table}
\caption{\label{sp:tab:cdf_comparison} Numerical comparison between the (empirical) CDF of the approximated distribution and the  theoretical CDF of the fiducial distribution.}
\centering
\vspace{0.2 in}
\begin{tabular}{lcccccc}
\hline
                & $\theta=1$ & $\theta=2$ & $\theta=3$ & $\theta=4$ & $\theta=5$ & $\theta=6$ \\ \hline
Estimated Probability   & 0.354      & 0.672      & 0.780      & 0.813      & 0.846      & 0.940      \\
Theoretical Probability & 0.361      & 0.679      & 0.783      & 0.815      & 0.845      & 0.940      \\ \hline
\end{tabular}
\end{table}


\begin{figure}[htbp]
\centering
\begin{minipage}{0.45\textwidth}
    \centering
    \includegraphics[width=\linewidth]{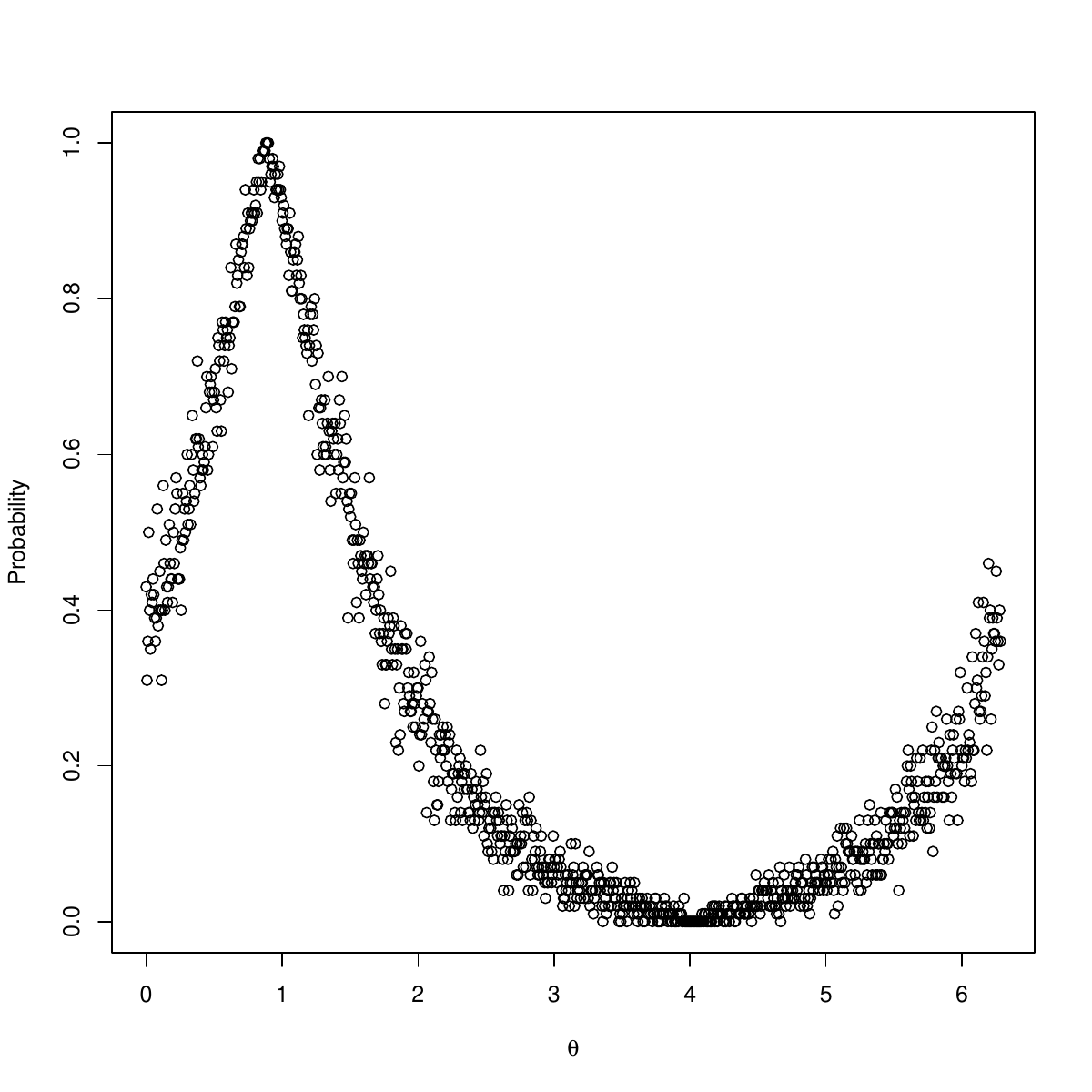} 
    \textbf{(a) \normalfont{$F_{\theta}(T_{y,\theta})$ as the function of $\theta$}}
\end{minipage}\hfill 
\begin{minipage}{0.45\textwidth}
    \centering
    \includegraphics[width=\linewidth]{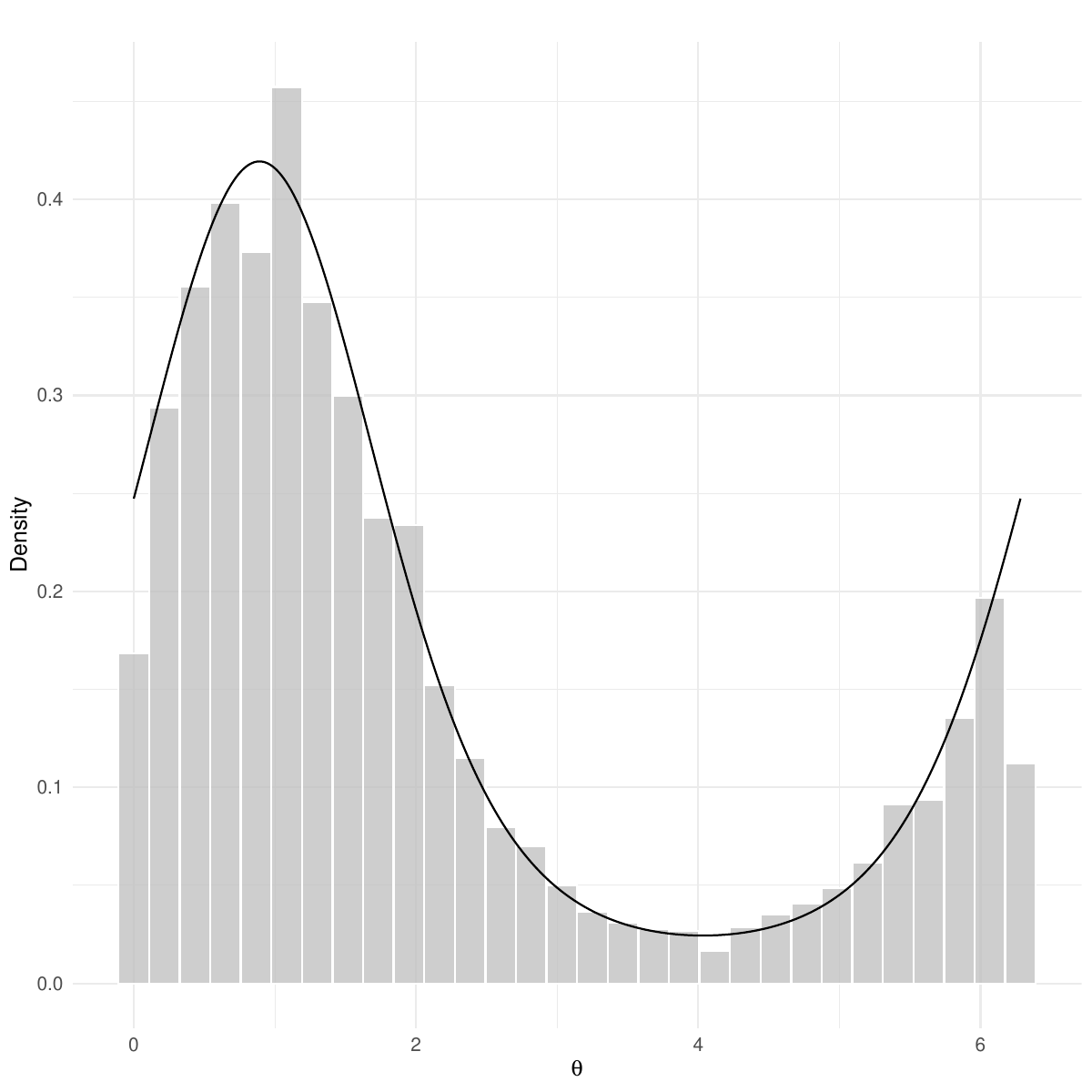} 
    \textbf{(b) \normalfont{Approximated Fiducial Distribution}}
\end{minipage}
\caption{The results of the DR algorithm on the direction measurements data. Panel (a) shows the estimated $F_{\theta}(T_{y,\theta})$ via Monte-Carlo (100 repetitions) as the function of $\theta$. Panel (b) is the histogram of the approximated fiducial distribution obtained with the DR algorithm. \label{sp:fig-von-fr}}
\end{figure}

\clearpage

\subsection{The Combined RA-DR Algorithm}\label{sps:alg-combined}
An effective algorithm combining the proposed RA and DR procedure is shown as Algorithm~\ref{sp:alg:combined}. This algorithm is applied to in the simulation study and real data analysis in Section~4.

The computational cost of the RA-DR algorithm is $O(B\cdot T \cdot |\boldsymbol{\alpha}|)$, where $|\boldsymbol{\alpha}|$ denotes the number of elements in the candidate set. In the simulation study, we use $B = 100$, $T = 100$, $|\boldsymbol{\alpha}| = 3$. As a result, the total computation cost is equivalent to $30,000$ bootstraps. Specifically, on a standard machine, the $\{n = 100, p = 30\}$ case takes about 1 minute, the $\{n = 200, p = 100\}$ case takes about 2 minutes and the $\{n = 500, p = 450\}$ case takes about 9 minutes.

\begin{algorithm}
\footnotesize
        Specify a set of target significance level $\boldsymbol{\alpha}$, \textit{e.g.} $\boldsymbol{\alpha} = \{0.1, 0.5, 0.9\}$ and 
        choose the number of bootstrap replications $B$ as well as the iteration number $T$\;
        Initialize an empty (indexed) parameter set $\tilde{\Theta}_0$ as well as an empty (indexed) function values set $\mathcal{F}_0$\;
        \For{$\alpha$ \textnormal{\textbf{in}} $\boldsymbol{\alpha}$} {
         Initialize $m^{(0)}$ with $m^{(0)}=n$\;
        \For{$t \gets 0$ \KwTo $T-1$} {
            Let $m=\lfloor m^{(t)}\rfloor + \mbox{Bernoulli}(m^{(t)} -\lfloor m^{(t)}\rfloor)$, where $\lfloor m^{(t)}\rfloor$ denotes the largest integer not greater than $m^{(t)}$ and $\mbox{Bernoulli}(p)$ a Bernoulli random variable with parameter $p$\;
            Sample  $\{(x_i^*, y_i^*):\; i=1,..., m^{(t)}\}$ or, in the matrix notation, $(X^*, Y^*)$ from $\{(x_i, y_i):\; i=1,..., n\}$ with replacement\;    
            Compute $\hat{\theta}_* = \argmin_\theta \ell(\theta, X^*, Y^*)$\;
            Evaluate $\ell^{(t)} = -[\ell(\hat{\theta}_*, X, Y) - 
            \ell(\hat{\theta}, X, Y)]$\;
            Initialize $P=0$\;
            \For{$b \gets 1$ \KwTo $B$}{
                Sample $\{(x_i^{**}, y_i^{**}):\; i=1,..., n\}$ or, in the matrix notation, $(X^{**}, Y^{**})$ from model $\mathbf{P}_{\hat{\theta}_*}$\;
                Compute $\hat{\theta}_{**} = \argmin_\theta \ell(\theta, X^{**}, Y^{**})$\;
                Evaluate $S^{(b)} = -[\ell(\hat{\theta}_{*}, X^{**}, Y^{**}) - \ell(\hat{\theta}_{**}, X^{**}, Y^{**})]$\;
                \If{$S^{(b)} \le \ell^{(t)}$}{
                    $P\gets P + 1$\;
                }
            }
            Set $Z_t \gets \frac{\mathbb{I}{\{P/B \le\alpha\}}  - \alpha}{\sqrt{B \alpha(1-\alpha)}}$\;
            Include $\hat{\theta}_*$ into $\tilde{\Theta}_0$ and $P/B$ into $\mathcal{F}_0$\;
            
            Update $m^{(t+1)} = m^{(t)} + \frac{c}{t + 1} Z_t$, where $c$ is a predefined constant\;
        }
        }
        Initialize a resampled parameter set $\tilde{\Theta}$ and conduct the refinement procedure in step (c) of Algorithm~2 using the elements in the set $\tilde{\Theta}_0$ and $\mathcal{F}_0$\;
   \caption{The RA-DR algorithm\label{sp:alg:combined} 
   }
\end{algorithm}

\clearpage

\bibliographystyle{apalike}
\bibliography{ims-supplementary}